\documentclass[aps,prx,twocolumn,superscriptaddress,longbibliography,10pt,floatfix]{revtex4-2}

\usepackage[utf8]{inputenc}
\usepackage[T1]{fontenc}
\usepackage{graphicx}
\graphicspath{{diagrams/}}
\usepackage{amsmath,amssymb,braket}
\usepackage{hyperref}
\usepackage{physics}        
\usepackage[dvipsnames]{xcolor}

\usepackage{float}              
\makeatletter                   
\let\newfloat\newfloat@ltx      
\makeatother                    
\usepackage{algorithm}
\usepackage{algorithmic}  
\usepackage{amsthm}

\usepackage{booktabs}

\usepackage{quantikz}

\usepackage{microtype} 
\usepackage{import}

\newtheorem{theorem}{Theorem}

\newtheorem{definition}{Definition}
\theoremstyle{remark}

\begin{document}

\title{Hybrid Quantum Error Correction and Mitigation by Purification}

\author{Jonathan Raghoonanan}
\affiliation{New York University Shanghai; NYU-ECNU Institute of Physics at NYU Shanghai, 567 West Yangsi Road, Shanghai, 200124, China.}
\affiliation{Department of Physics, New York University, New York, NY 10003, USA}
\author{Tim Byrnes}
\email{tim.byrnes@nyu.edu}
\affiliation{New York University Shanghai; NYU-ECNU Institute of Physics at NYU Shanghai, 567 West Yangsi Road, Shanghai, 200124, China.}
\affiliation{Center for Quantum and Topological Systems (CQTS), NYUAD Research Institute, New York University Abu Dhabi, UAE}
\affiliation{Department of Physics, New York University, New York, NY 10003, USA}

\date{\today}

\begin{abstract}

Quantum error correction physically removes errors from a quantum state, while quantum error mitigation improves observable estimates by processing noisy measurement data. We introduce \emph{purification quantum error suppression} (PQES), a hybrid approach that uses multiple noisy copies of an unknown state to combine these two ideas. The protocol uses SWAP tests to physically reduce errors by purification, while the full outcome record is used to combine all branches without postselection. In this way, PQES avoids the fixed-success-outcome requirement of standard SWAP-test purification while still accessing the power-purified state $\rho^N$. The SWAP identities allow purification steps to be interleaved with unitary circuit blocks, so errors can be suppressed during a computation rather than only at the final measurement. We provide both a parallel binary-tree implementation and a more compact register-recycled implementation using $O(M\ell)$ coherent data qubits for an $M$-qubit register and $N=2^\ell$ input copies. We analyze the resulting error thresholds under representative noise models. For local depolarizing noise on the product-state family studied here, the threshold is $p_{\mathrm{th}}=3/4$ for any register size, while local dephasing of $|+\rangle^{\otimes M}$ has a threshold of $p_{\mathrm{th}}=1/2$. Local Clifford twirling can be used to convert dephasing into a depolarization channel and restore the higher threshold.
\end{abstract}

\maketitle

\section{Introduction}
\label{sec:intro}


Quantum computation requires mechanisms for suppressing errors before noise overwhelms the useful circuit dynamics. The best-developed route is quantum error correction (QEC), where information is encoded into a larger Hilbert space, errors map states into distinguishable codespaces, and syndrome measurements guide recovery operations to return states back to the logical code space~\cite{NielsenChuang2010,devitt2013quantum,Watrous2018,WoodGambetta2018PTM}. Stabilizer codes provide the dominant framework for scalable fault tolerance because of their locality and high circuit-level thresholds~\cite{Gottesman1997Thesis,Fowler_2012}. However, their power comes with substantial architectural overhead: repeated syndrome extraction, real-time decoding, magic-state factories for non-Clifford gates, and stringent requirements on gates, measurements, resets, and connectivity. These costs motivate complementary error-suppression strategies that exploit complementary physical resources toward enabling larger-scale and higher-depth quantum circuits. Several such strategies are already central to quantum control and near-term computation. Early examples include decoherence-free subspaces and noiseless subsystems, which protect information by encoding it into symmetry sectors that are insensitive to specific collective noise models~\cite{Lidar2012DFSReview,Dankert2009Designs,WallmanEmerson2016}. Dynamical decoupling suppresses unwanted system--bath couplings through fast control pulses~\cite{Viola1999DD}. More recently, quantum error mitigation (QEM) techniques instead leave the noisy state largely intact and reconstruct improved expectation values from noisy data, for example by noise extrapolation or quasiprobability cancellation~\cite{Temme2017Mitigation,Cai2023QEM}. These approaches illustrate a broader design space: one may suppress errors by encoding, by control, by post-processing, or by consuming additional quantum resources.

A closely related line of work, traditionally framed in the context of long-distance communication and resource preparation rather than computation, is entanglement purification \cite{Dur2007EPReview,yan2023advances}. In a typical purification protocol, one starts with several noisy copies of a quantum state, and using only local operations and classical communication (LOCC), a single higher fidelity copy of the target state is obtained after post-selection.  
In most purification protocols, the \emph{target state is fixed and known}, such as a specific Bell pair, GHZ state, or graph state. Furthermore, they typically employ postselection where particular measurement outcomes are desirable and otherwise the resultant state is discarded. For example, the original protocols of Bennett, Wootters, and co-workers distilled high-fidelity Bell pairs from many noisy copies, together with  twirling operations to an isotropic Bell-diagonal form~\cite{Bennett1996Twirling}. Deutsch, Sanpera, and co-workers incorporated such purification steps into quantum privacy amplification for quantum key distribution over noisy channels~\cite{Deutsch1996QPA}. Subsequent work generalized these ideas to multipartite entangled states~\cite{Murao1998, Saeed2025}. Purification is considered essential for quantum repeaters and networking architectures~\cite{Pan2001EPComm,Briegel1998,Kimble_2008}. Modern variants optimize gate counts, success probabilities, and robustness for concrete network architectures~\cite{Torres2024Purif,Illiano_2022}, but the basic setting remains one in which a known resource state is distilled from many noisy copies.

Historically, purification has been mainly considered in the context of quantum state preparation due to the fact that the target state is fixed and known.  However, a more general \emph{unknown-state} purification can be performed, where the task is to recover a nearly pure copy of an unknown state from many noisy copies, without prior knowledge of which pure state was prepared. This can be performed when typical restrictions, such as LOCC-only processing, are relaxed. Early attempts include optimal collective processing and universal purification procedures, but they are generally limited to single qubits, asymptotic regimes, and may require limited knowledge of the initial state~\cite{Cirac_1999,keyl_1999}. More recently, Childs and co-workers showed that a recursive SWAP-test-based protocol can boost the fidelity of an unknown pure state given multiple noisy copies for qudits of arbitrary dimension~\cite{Childs_2025}. Here, the SWAP test consists of a quantum primitive that determines the fidelity between two states \cite{Buhrman2001SwapTest}. In the purification context, by post-selecting on the symmetric sector of the test, one can iteratively amplify overlap with the dominant eigenvector of the input density operator. Because the protocol acts on whatever pure state occupies the dominant eigenvalue, it naturally aligns with the idea of protecting unknown encoded information rather than a fixed resource state. This is the purification analogue of syndrome extraction: a measurement partitions the enlarged Hilbert space into sectors carrying different information about the noise.

In parallel, purification-based quantum error mitigation (QEM) methods exploit the fact that the ideal output of a quantum computation is often approximately pure. A prominent example is virtual distillation, introduced by Huggins, McClean, and co-workers~\cite{Huggins2021}. In virtual distillation, one estimates expectation values with respect to the normalized power state $\rho^N/\Tr(\rho^N)$, even though the physical state available in the laboratory is $\rho^{\otimes N}$. Koczor independently developed a closely related error-suppression framework based on derangement operators~\cite{Koczor_derangement}, and further clarified the dominant-eigenvector limitation: the power map converges to the dominant eigenvector of the noisy state, which may differ from the ideal noiseless state~\cite{Koczor2021_dominantEV}. Subsequent work has improved the resource profile of these techniques, including reset-based qubit-efficient virtual distillation~\cite{czarnik_qubit_eff} and constant-depth multivariate trace-estimation circuits~\cite{Quek_2024_const_depth}. This raises a natural question. Conventional QEC acts on the quantum state itself: syndrome information is used to remove physical errors so that the corrected state can continue through the circuit. QEM, by contrast, usually acts at the level of measurement data: noisy outcomes from one or more circuits are combined to estimate a less noisy observable. Is there a useful middle ground, where multiple noisy copies are used to physically suppress errors in the quantum registers while still retaining the favorable no-postselection logic of QEM?

\begin{figure*}[t]
  \includegraphics[width=\linewidth]{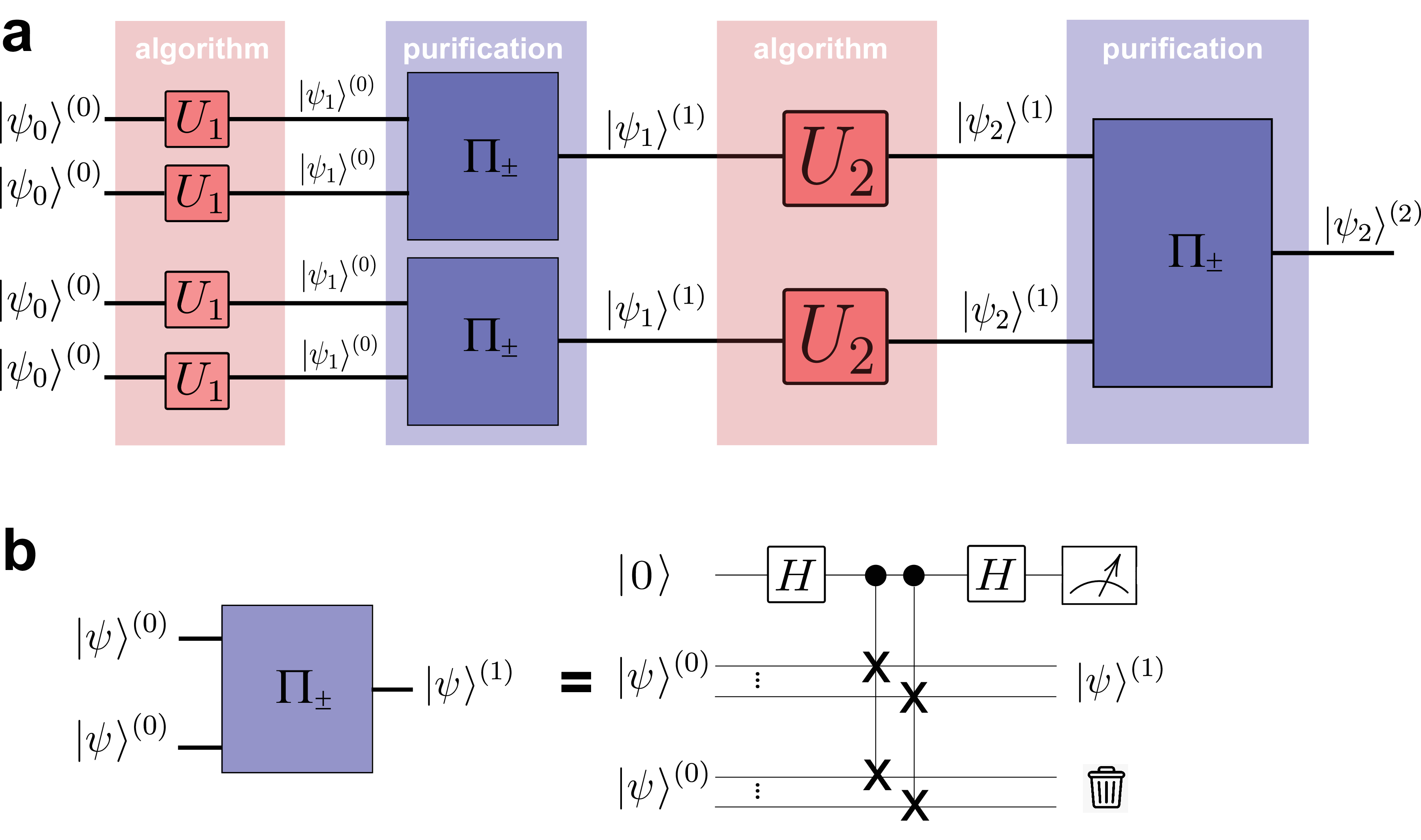}
  \caption{Purification quantum error suppression implementation in a quantum circuit. (a) Binary-tree configuration of SWAP-test purification steps interleaved within an $M$-qubit quantum algorithm consisting of a sequence of unitary gates $U_n$. The full quantum circuit is $U=\prod_n U_n$. The label $\ket{\psi_n}^{(\ell)}$ denotes the register state at the $n$th stage of the algorithm after $\ell$ rounds of SWAP-test purification processing. In the non-postselected PQES mode, the branch state continues through subsequent circuit layers, while the total SWAP parity selects observables of the power-purified state $\rho^{2^\ell}/\Tr(\rho^{2^\ell})$. (b) The SWAP gadget and its definition. The $\Pi_+$ measurement corresponds to the ancilla outcome $|0\rangle$ and $\Pi_-$ corresponds to $|1\rangle$. The second copy of the state is discarded or traced out. Thick horizontal lines denote $M$-qubit registers, while thin horizontal lines denote single qubits.}
  \label{diag:main_use_case}
\end{figure*}

In this paper, we answer this question using SWAP-test purification. We call the resulting scheme \emph{purification quantum error suppression} (PQES). The protocol uses multiple noisy copies of the same unknown state. Pairs of copies are combined using SWAP tests, and the surviving registers are physically transformed by the measurement outcomes. If one keeps only the symmetric SWAP outcome, one recovers the standard postselected SWAP-purification protocol~\cite{Childs_2025}. The key difference in PQES is that there is no postselection, i.e. no SWAP outcome is discarded. Symmetric and antisymmetric outcomes are both retained, and the outcome record is used to combine the branches with the correct signs. This gives a hybrid QEC/QEM structure. Like QEC, PQES physically changes the quantum state during a repeated noise--recovery cycle, so the surviving register can continue through later circuit layers. Like QEM, it uses multiple noisy copies and an outcome-dependent estimator to access the power-purified state $\rho^N/\Tr(\rho^N)$. The SWAP identities also allow unitary gates to be interleaved with purification steps, so the protocol can suppress errors during the circuit rather than only at the end.

The performance of PQES is governed by whether the desired target state remains the dominant eigenvector of the noisy density matrix after each error cycle. If it does, increasing the number of purification rounds amplifies the target component; if it does not, purification converges to the wrong state. This gives a natural generalization of the concept of an error threshold that can be defined in the context of PQES. For the local depolarizing model studied here, this threshold is $p_{\mathrm{th}}=3/4$ for any register size. We also show that local dephasing has a lower threshold because it changes the eigenbasis of the noisy state, and that twirling can restore depolarizing behavior.



\section{SWAP-based purification}
\label{sec:swap-basics}

In this section, we review the SWAP test and the branch maps that underlie both postselected SWAP-based purification and the non-postselected PQES protocol.

\subsection{The SWAP test}
\label{subsec:swap-basics}

The core operation is the SWAP test applied to two registers, which may contain either identical copies of a noisy state or two intermediate states generated by the same recursive procedure \cite{Childs_2025}. Let $A,B$ be two quantum registers, each with a $D$-dimensional Hilbert space, where $D=2^M$ for $M$ qubits. In the SWAP test, as shown in Fig.~\ref{diag:main_use_case}(b), an ancilla qubit controls the SWAP operation on the two registers. The measurement outcomes on the ancilla partition the $AB$ registers into symmetric and antisymmetric subspaces. The $ | 0 \rangle $ outcome projects $AB$ onto the symmetric subspace
\begin{equation}
\Pi_{+} \ket{\psi}_A \otimes \ket{\phi}_B = \frac{1}{2} \left( \ket{\psi}_A \otimes \ket{\phi}_B+\ket{\phi}_A \otimes \ket{\psi}_B \right) ,
\label{eq:piplussym}
\end{equation}
while $ |1 \rangle $ projects onto the antisymmetric subspace,
\begin{equation}
\Pi_{-} \ket{\psi}_A \otimes \ket{\phi}_B = \frac{1}{2} \left( \ket{\psi}_A \otimes \ket{\phi}_B - \ket{\phi}_A \otimes \ket{\psi}_B \right) .
\end{equation}
These are unnormalized projected states. The projectors onto the symmetric and antisymmetric subspaces are
\begin{equation}
\Pi_{\pm}=\tfrac12\bigl(I\pm \mathrm{SWAP}\bigr),
\label{pidef}
\end{equation}
where the register-SWAP operator is
%
%
\begin{align}
\mathrm{SWAP} &= \prod_{j=1}^{M}\mathrm{SWAP}_{j}, \nonumber\\
\mathrm{SWAP}_{j} &=
\frac{1}{2}
\left( I^A_j I^B_j + X^A_j X^B_j + Y^A_j Y^B_j + Z^A_j Z^B_j \right).
\label{eq:registerswap}
\end{align}
Here $X^A_j,Y^A_j,Z^A_j$ act on qubit $j$ of register $A$, and $X^B_j,Y^B_j,Z^B_j$ act on the corresponding qubit of register $B$. The projectors satisfy $\Pi_+ + \Pi_-=I $ and $ \Pi_\pm^2=\Pi_\pm $. For $M$-qubit registers, the controlled register-SWAP decomposes into $M$ Fredkin gates, one for each physical qubit pair $A_jB_j$~\cite{Buhrman2001SwapTest,NielsenChuang2010,Watrous2018}. These Fredkin gates can be parallelized when the architecture supplies a suitable fanout/GHZ control or a native multi-target controlled-SWAP; otherwise this decomposition has depth $O(M)$ in the standard circuit model.

\subsection{SWAP gadget}

In the analysis that follows, we will make use of the SWAP gadget, which involves performing the SWAP test projection (\ref{pidef}), then discarding the second register. The resulting state for an initial product state $ \rho \otimes \rho' $ is \cite{Buhrman2001SwapTest,Childs_2025}
\begin{align}
\rho_{\pm} = & \frac{\Tr_B \left[\Pi_{\pm}(\rho\otimes\rho')\Pi_{\pm}\right]}
    {P_{\pm} } \nonumber \\
     = &  \frac{1} {4 P_{\pm}  } \Big( \Tr_B [\rho\otimes\rho']+ \Tr_B [\rho'\otimes\rho] \nonumber \\
     & \pm \Tr_B [\mathrm{SWAP}(\rho\otimes\rho')] \pm \Tr_B [(\rho\otimes\rho')\mathrm{SWAP}]  \Big)    
     \nonumber\\
     = & \frac{\rho + \rho' \pm (\rho \rho' + \rho' \rho ) }{4 P_{\pm}} .    
     \label{eq:generalswapgadget}
\end{align}
In obtaining the last line, we used the partial-trace SWAP identities
\begin{align}
\Tr_B[(\rho\otimes\rho')\mathrm{SWAP}] &= \rho\rho',
\nonumber\\
\Tr_B[\mathrm{SWAP}(\rho\otimes\rho')] &= \rho'\rho .
\label{eq:partialswaptrick}
\end{align}
The probabilities of the two outcomes are
\begin{align}
P_{\pm} = \frac{1}{2} ( 1 \pm \Tr [ \rho \rho']) .
\label{eq:swapprobs}
\end{align}
For identical inputs, $\rho'=\rho$, (\ref{eq:generalswapgadget}) reduces to
\begin{align}
\rho_{\pm} = \frac{\rho \pm \rho^2}{1\pm \Tr(\rho^2)},
\label{eq:identicalswapbranches}
\end{align}
with corresponding probability
\begin{align}
P_{\pm} = \frac{1}{2}\left(1\pm \Tr(\rho^2)\right).
\label{eq:identicalswapbranches_prob}
\end{align}
%

%
%



\section{Purification quantum error suppression}
\label{sec:PQES}



We now construct PQES from the SWAP-gadget identities of Sec.~\ref{sec:swap-basics}. We first explain the single-round purification mechanism, then show how the purified component can be obtained without postselection by using the SWAP outcome record. We then generalize the construction to multiple rounds and explain how to extract observables.

\subsection{Single-round purification}

The basic purification effect of the SWAP gadget follows by setting the two input states to be identical, $\rho'=\rho$, in ~(\ref{eq:generalswapgadget}). The two resulting branch states are given in~(\ref{eq:identicalswapbranches}). For the $+$ outcome, the resulting state is a physical mixture of the original state $\rho$ and the squared state $\rho^2$. The squared state may be considered a purified version of $\rho$ because, if
\begin{align}
\rho=\sum_i \lambda_i |\lambda_i\rangle\langle \lambda_i|,
\end{align}
then normalization after squaring gives
\begin{align}
{\cal P}(\rho) := \frac{\rho^2}{\Tr(\rho^2)} =
\frac{\sum_i \lambda_i^2 |\lambda_i\rangle\langle \lambda_i|}
{\sum_i \lambda_i^2}.
\label{eq:purifiedstaten2}
\end{align}
Thus, the relative weight of eigencomponents with smaller eigenvalues is suppressed. This is the basic spectral mechanism underlying virtual distillation and error suppression by derangements \cite{Huggins2021,Koczor_derangement}.

Furthermore, using (\ref{eq:identicalswapbranches})-(\ref{eq:identicalswapbranches_prob}), we see that if the two branch states are averaged with their ordinary probabilities, then the purification terms cancel and simply return the original state:
\begin{align}
P_+\rho_+ + P_-\rho_- = \rho.
\label{eq:ordinarybranchaverage}
\end{align}
However, if we instead take the signed branch average of the SWAP test outcomes, we can actually obtain the square of the density matrix,
\begin{align}
P_+\rho_+ - P_-\rho_- = \rho^2 .
\label{eq:signedbranchaverage}
\end{align}
Equivalently, for any observable $O$,
\begin{align}
P_+\Tr(O\rho_+) - P_-\Tr(O\rho_-) = \Tr(O\rho^2).
\label{eq:signedobservableidentity}
\end{align}
This distinction is the basic physical mechanism of PQES. Each $\rho_\pm$ is a genuine conditional state of the surviving register. Keeping only the $+$ outcome, or repeating the SWAP test until the $+$ outcome occurs, gives the postselected physical purification mode used in streaming SWAP-based state purification~\cite{Childs_2025}. The PQES construction used here instead keeps both outcomes. Though the signed average state does not appear to be a physical state in the same way that the average state is, the relative sign in \eqref{eq:signedbranchaverage} is not an arbitrary post-processing rule: it is the eigenvalue of the SWAP outcome recorded by the ancilla. Thus, the signed average can be implemented as the measurement of a physical observable by combining the outcome record with the system state.


The signed average can be equivalently implemented using a physical state and conventional expectation values using the following approach. If the ancilla or classical SWAP-test record is retained, the resolved state may be written as
\begin{align}
\widetilde{\rho}
& = P_+ |0 \rangle \langle 0 | \otimes \rho_+  + P_- |1 \rangle \langle 1 | \otimes \rho_-  \nonumber \\
& = \frac{1}{2} (  I \otimes \rho +  Z \otimes \rho^2 ) .
\end{align}
For any observable $O$ on the system register,
\begin{align}
\langle I \otimes O \rangle &  = \Tr ( I \otimes O \widetilde{\rho} ) = \Tr ( O \rho ) , \label{line1ioeq}  \\
\langle Z \otimes O \rangle &  = \Tr ( Z \otimes O \widetilde{\rho} ) = \Tr ( O \rho^2 )  . \label{line2ioeq}
\end{align}
Hence, the $Z$-weighted record picks out the purified component, which is encoded as a sector of the enlarged system. In practice, one may either keep the ancilla coherently until the relevant measurement, or measure it and store the corresponding classical sign as in (\ref{eq:signedbranchaverage}). These two descriptions give the same estimator for $\Tr(O\rho^2)$. For compactness, we henceforth use the latter implementation. It should be emphasized, however, that the same quantity may be obtained directly from the physically enlarged state by measuring the observable $Z\otimes O$, as in ~\eqref{line2ioeq}. It follows, then, that the same SWAP primitive supports two uses: In the postselected mode, one obtains a physically purified branch; in the non-postselected mode, one obtains physical branch dynamics together with signed access to expectation values of the purified state $\rho^2/\Tr(\rho^2)$. The multi-round protocol below generalizes (\ref{eq:signedbranchaverage}) from $\rho^2$ to higher powers $\rho^{2^\ell}$.

\subsection{Multiple purification rounds}

We now generalize the single-round identity to multiple purification rounds, as illustrated in Fig.~\ref{diag:main_use_case}(a). For the moment, we disregard the algorithmic part of the circuit and set all unitary blocks $U_n$ to the identity. After $\ell$ rounds of binary-tree purification, there are $2^\ell-1$ SWAP-test outcomes. We denote the outcome string by
\begin{align}
    \vec{\sigma}_\ell = (\sigma_1, \sigma_2, \dots , \sigma_{2^{\ell} - 1 }  ) ,
\end{align}
where $\sigma_i=\pm1$, with $+1$ corresponding to the symmetric branch and $-1$ corresponding to the antisymmetric branch. The conditional state associated with the outcome string $\vec{\sigma}_\ell$ is denoted $\rho_{\vec{\sigma}_\ell}$, and its probability is denoted $P_{\vec{\sigma}_\ell}$. These probabilities are defined recursively. For one round,
\begin{align}
P_{\sigma_1}=\frac{1}{2}\left(1+\sigma_1\Tr(\rho^2)\right).
\end{align}
If two subtrees have outcome strings $\vec{\alpha}$ and $\vec{\beta}$, probabilities $P_{\vec{\alpha}}$ and $P_{\vec{\beta}}$, and conditional output states $\rho_{\vec{\alpha}}$ and $\rho_{\vec{\beta}}$, then the probability of the merged branch $(\vec{\alpha},\vec{\beta},\sigma)$ is
\begin{align}
P_{(\vec{\alpha},\vec{\beta},\sigma)}
=
P_{\vec{\alpha}}P_{\vec{\beta}}\,
\frac{1}{2}\left(1+\sigma\Tr[\rho_{\vec{\alpha}}\rho_{\vec{\beta}}]\right).
\label{eq:recursivebranchprob}
\end{align}
At each internal node, two branch-dependent input states are combined by the SWAP gadget. For the two input states $\rho_{\vec{\alpha}}$ and $\rho_{\vec{\beta}}$, (\ref{eq:generalswapgadget}) gives
\begin{align}
\rho_{(\vec{\alpha},\vec{\beta},\sigma)}
=
\frac{\rho_{\vec{\alpha}}+\rho_{\vec{\beta}}
+\sigma(\rho_{\vec{\alpha}}\rho_{\vec{\beta}}+\rho_{\vec{\beta}}\rho_{\vec{\alpha}})}
{4P_{\sigma|\vec{\alpha},\vec{\beta}}},
\label{rhorecursive}
\end{align}
where
\begin{align}
P_{\sigma|\vec{\alpha},\vec{\beta}} =
\frac{1}{2}\left(1+\sigma\Tr[\rho_{\vec{\alpha}}\rho_{\vec{\beta}}]\right).
\label{eq:P_sigma_given_ab}
\end{align}
This recursion implies that every conditional branch state is a polynomial in the original input state $\rho$, divided by a scalar normalization factor. More precisely, after $\ell$ rounds, each branch state is a polynomial in $\rho$ of degree at most $2^\ell$. This follows by induction: at $\ell=1$, (\ref{eq:generalswapgadget}) contains only $\rho$ and $\rho^2$; if two input branch states are polynomials in $\rho$ of degree at most $2^\ell$, then (\ref{rhorecursive}) forms a linear combination of those polynomials and their products, whose degree is at most $2^{\ell+1}$.

For example, after $\ell=2$ rounds starting from $\rho^{\otimes 4}$, let $\sigma_1$ and $\sigma_2$ be the first two SWAP outcomes and let $\sigma_3$ be the final SWAP outcome. It is clearest to write the unnormalized branch state,
\begin{align}
P_{\sigma_1\sigma_2\sigma_3}\rho_{\sigma_1\sigma_2\sigma_3}
=
\frac{1}{8}\Big[
&(P_{\sigma_1}+P_{\sigma_2})\rho
+(\sigma_1P_{\sigma_2}+\sigma_2P_{\sigma_1})\rho^2 \nonumber\\
&+\sigma_3\rho^2
+\sigma_3(\sigma_1+\sigma_2)\rho^3
+\sigma_1\sigma_2\sigma_3\rho^4
\Big],
\label{eq:l2_unnormalized_branch}
\end{align}
where $P_{\sigma_i}=\frac12(1+\sigma_i\Tr\rho^2)$ for the two lower branches, and $P_{\sigma_1\sigma_2\sigma_3}$ denotes the full probability of the complete three-outcome branch. Similar to the $ \ell = 1 $ case of (\ref{eq:ordinarybranchaverage}), taking an average of all the outcomes simply produces the original state
\begin{align}
\sum_{\sigma_1 \sigma_2 \sigma_3} P_{\sigma_1\sigma_2\sigma_3} \rho_{\sigma_1 \sigma_2 \sigma_3 }  = \rho . 
\end{align}
Note that all other terms evaluate to zero since $ \sum_{\sigma_i} \sigma_i = 0 $. We may extract the higher powers of $ \rho $ by multiplying by suitable factors of $ \sigma_i $ and using the fact that $ \sigma^2_i =1 $. The important feature is the last term: the highest power $\rho^4$ is tagged by the total parity $\sigma_1\sigma_2\sigma_3$. Therefore,
\begin{align}
\sum_{\sigma_1,\sigma_2,\sigma_3}
P_{\sigma_1\sigma_2\sigma_3}\,
\sigma_1\sigma_2\sigma_3\,
\rho_{\sigma_1\sigma_2\sigma_3}
=
\rho^4 .
\end{align}
All lower powers vanish under the parity sum. The key point is that the highest power can be isolated by multiplying each branch by the total parity of all SWAP outcomes,
\begin{align}
    \Omega_{\vec{\sigma}_{\ell}} = \prod_{i=1}^{2^\ell-1}\sigma_i .
\end{align}
Here, $\Omega_{\vec{\sigma}_{\ell}}=+1$ if an even number of antisymmetric outcomes occurred and $\Omega_{\vec{\sigma}_{\ell}}=-1$ if an odd number occurred.

\begin{theorem}[]
For $2^\ell$ identical input copies of $\rho$ arranged in the binary-tree SWAP-test circuit, the ordinary branch average gives
\begin{align}
\sum_{\vec{\sigma}_{\ell}} P_{\vec{\sigma}_{\ell}} \rho_{\vec{\sigma}_{\ell} } = \rho ,
\label{avstategen}
\end{align}
while the total-parity-weighted branch average gives
\begin{align}
\sum_{\vec{\sigma}_{\ell}}  P_{\vec{\sigma}_{\ell}} \Omega_{\vec{\sigma}_{\ell}}    \rho_{\vec{\sigma}_{\ell} } = \rho^{2^{\ell}} .
\label{purifiedrhol}
\end{align}
\end{theorem}

\begin{proof}
We first record the one-step identities for arbitrary input states $X$ and $Y$. From (\ref{eq:generalswapgadget}),
\begin{align}
P_{\sigma}\rho_{\sigma}(X,Y) = \frac{1}{4} \left[ X+Y+\sigma(XY+YX) \right],
\end{align}
and hence
\begin{align}
\sum_{\sigma=\pm1}P_{\sigma}\rho_{\sigma}(X,Y)
&=
\frac{X+Y}{2},
\label{onestepordinary}\\
\sum_{\sigma=\pm1}\sigma P_{\sigma}\rho_{\sigma}(X,Y)
&= \frac{XY+YX}{2}.
\label{onestepsigned}
\end{align}
Equation~(\ref{onestepordinary}) proves (\ref{avstategen}) by induction: at each internal node, the ordinary average of the parent is the average of the ordinary averages of its two children, and all leaves are equal to $\rho$.

We now prove (\ref{purifiedrhol}). For $\ell=1$, (\ref{onestepsigned}) with $X=Y=\rho$ gives
\begin{align}
\sum_{\sigma=\pm1}\sigma P_\sigma\rho_\sigma(\rho,\rho)=\rho^2,
\end{align}
which is (\ref{eq:signedbranchaverage}). Suppose that, after $\ell$ rounds, the signed branch average of each subtree is $\rho^{2^\ell}$. At the next level, the total parity is the product of the left-subtree parity, the right-subtree parity, and the parent outcome $\sigma$. Using (\ref{onestepsigned}) and summing over the left and right subtree outcomes gives
\begin{align}
\sum_{\vec{\sigma}_{\ell+1}} P_{\vec{\sigma}_{\ell+1}} \Omega_{\vec{\sigma}_{\ell+1}} \rho_{\vec{\sigma}_{\ell+1}} &=
\frac{1}{2} \left[ \left(\rho^{2^\ell}\right)\left(\rho^{2^\ell}\right) + \left(\rho^{2^\ell}\right)\left(\rho^{2^\ell}\right) \right] \nonumber\\
&= \rho^{2^{\ell+1}} .
\end{align}
This completes the induction. 
\end{proof}

Eq.~(\ref{purifiedrhol}) is the central SWAP-outcome identity. The highest power $\rho^{2^\ell}$ appears with unit coefficient in the parity-weighted branch average, just as $\rho$ appears with unit coefficient in the ordinary branch average of (\ref{avstategen}). The purification power is therefore not produced by exponentially small amplitudes in the branch algebra. The practical overhead enters through the number of consumed copies and through the variance of the signed estimator. As the normalized version of the $\ell$-fold purified state, we define
\begin{align}
{\cal P}_{\ell} (\rho) := \frac{\rho^N}{\Tr ( \rho^N)} = \frac{ \sum_i \lambda_i^N  |\lambda_i \rangle \langle \lambda_i | }
{\sum_i \lambda_i^N},
\qquad N=2^\ell .
\label{eq:purifiedstate}
\end{align}
This is the same normalized power state accessed in virtual distillation and derangement-based error suppression~\cite{Huggins2021,Koczor_derangement}. The formulation developed here adds the physical sector picture: $\rho^N$ is selected by an explicit SWAP-outcome parity record generated from pairwise branch maps. This representation is useful for streaming copies, register recycling, interleaving with unitary circuit blocks, and repeated noise--purification cycle analysis, as we show later in Sec.~\ref{sec:global_depol}-\ref{sec:local_anisotropic_errors}. Other implementations of related multivariate trace-estimation tasks include reset-based qubit-efficient constructions and constant-depth cyclic-shift estimators \cite{czarnik_qubit_eff,Quek_2024_const_depth}.

\subsection{Extracting observables}
\label{sec:observables}

We now discuss how observables are evaluated from the purified state. For any observable $O$, the expectation value with respect to (\ref{eq:purifiedstate}) is
\begin{align}
    \langle O \rangle_\ell = \frac{\Tr ( O \rho^N)}{\Tr (  \rho^N)}.
    \label{oexp}
\end{align}
Substituting (\ref{purifiedrhol}) into (\ref{oexp}) gives
\begin{align}
    \langle O \rangle_\ell
    =
    \frac{\sum_{\vec{\sigma}_{\ell}} P_{\vec{\sigma}_{\ell}} \Omega_{\vec{\sigma}_{\ell}}  \langle O \rangle_{\vec{\sigma}_{\ell} } }
    {\sum_{\vec{\sigma}_{\ell}} P_{\vec{\sigma}_{\ell}} \Omega_{\vec{\sigma}_{\ell}}   } ,
    \label{purifiedoexp}
\end{align}
where
\begin{align}
 \langle O \rangle_{\vec{\sigma}_{\ell} } := \Tr ( O \rho_{\vec{\sigma}_{\ell} })
 \label{standardoexp}
\end{align}
is the expectation value of $O$ with respect to the conditional branch state. The denominator in (\ref{purifiedoexp}) is
\begin{align}
\sum_{\vec{\sigma}_{\ell}} P_{\vec{\sigma}_{\ell}} \Omega_{\vec{\sigma}_{\ell}} =
\Tr(\rho^N),
\end{align}
which is the normalization factor of the power-purified state.

From (\ref{purifiedoexp}), we may estimate the sampling overhead associated with PQES. In addition to the statistics required to obtain the branch expectation values in (\ref{standardoexp}), one also needs statistics to estimate both the numerator and denominator in (\ref{purifiedoexp}). The standard error of (\ref{purifiedoexp}) is given by (see Appendix~\ref{app:overhead})
\begin{align}
  \epsilon & \approx \frac{ \sqrt{ \mathrm{Var}\!\left[ \Omega_{\vec{\sigma}_{\ell}} \left( \langle O \rangle_{\vec{\sigma}_{\ell}} - \langle O \rangle_\ell \right) \right]}} {\sqrt{N_\text{samp}} \Tr (\rho^N)},
  \label{standarderror}
\end{align}
where $N_\text{samp}$ is the number of samples, and the variance is taken with respect to the probability distribution $P_{\vec{\sigma}_{\ell}}$. For the case that $O$ is normalized so that $\|O\|_\infty\le 1$, such as for Pauli strings, the numerator of (\ref{standarderror}) can be bounded above by $2$ (see Appendix~\ref{app:overhead}). Rearranging for $N_{\mathrm{samp}}$ gives the sampling requirement
\begin{equation}
\label{simplifiedstandarderror}
N_{\mathrm{samp}} \gtrsim  \frac{4}{\epsilon^2 [\Tr(\rho^N)]^2}.
\end{equation}
We see that there is an additional sampling overhead that increases with the impurity of the state and with the purification power $N$. For perfectly pure input states, only the symmetric SWAP outcome occurs and $\Tr(\rho^N)=1$. For highly mixed states, $\Tr(\rho^N)$ can become small, increasing the number of samples required to resolve the signed estimator.

\subsection{Purifying quantum circuits}

Up to this point, we have not considered the algorithmic part of the purification circuit, represented by the unitary blocks in Fig.~\ref{diag:main_use_case}(a). The recursive SWAP-test construction is naturally compatible with unitary evolution applied identically to all live copies. If a unitary block $U_n$ is applied in parallel to each copy before a SWAP layer, then each input state is updated as
\begin{align}
\rho \rightarrow \rho' =  U_n \rho  U_n^\dagger  .
\end{align}
The branch identities above then hold with the replacement $\rho\rightarrow \rho'$. In particular,
\begin{align}
{\cal P}_{\ell}(U_n\rho U_n^\dagger)
=
U_n {\cal P}_{\ell}(\rho) U_n^\dagger ,
\end{align}
because $(U_n\rho U_n^\dagger)^N=U_n\rho^N U_n^\dagger$. Thus, for ideal unitary blocks, the purification operation commutes with the algorithmic evolution. This allows the quantum algorithm to be incorporated into the binary-tree purification architecture. One may either prepare multiple noisy copies of the final circuit output and apply the recursive SWAP-test estimator at the end, or interleave common unitary blocks between purification layers so that the copies are updated in parallel as the computation proceeds. The interleaved form is especially natural when copies are produced or refreshed in a streaming architecture, or when one wishes to estimate power-purified observables of selected intermediate states as well as the final output. In this sense, the method provides a flexible purification-QES architecture: the same signed-branch identities apply at any stage where the SWAP inputs represent corresponding noisy copies of the same target state.




\begin{figure}[t]
  \centering
  \includegraphics[width=\linewidth]{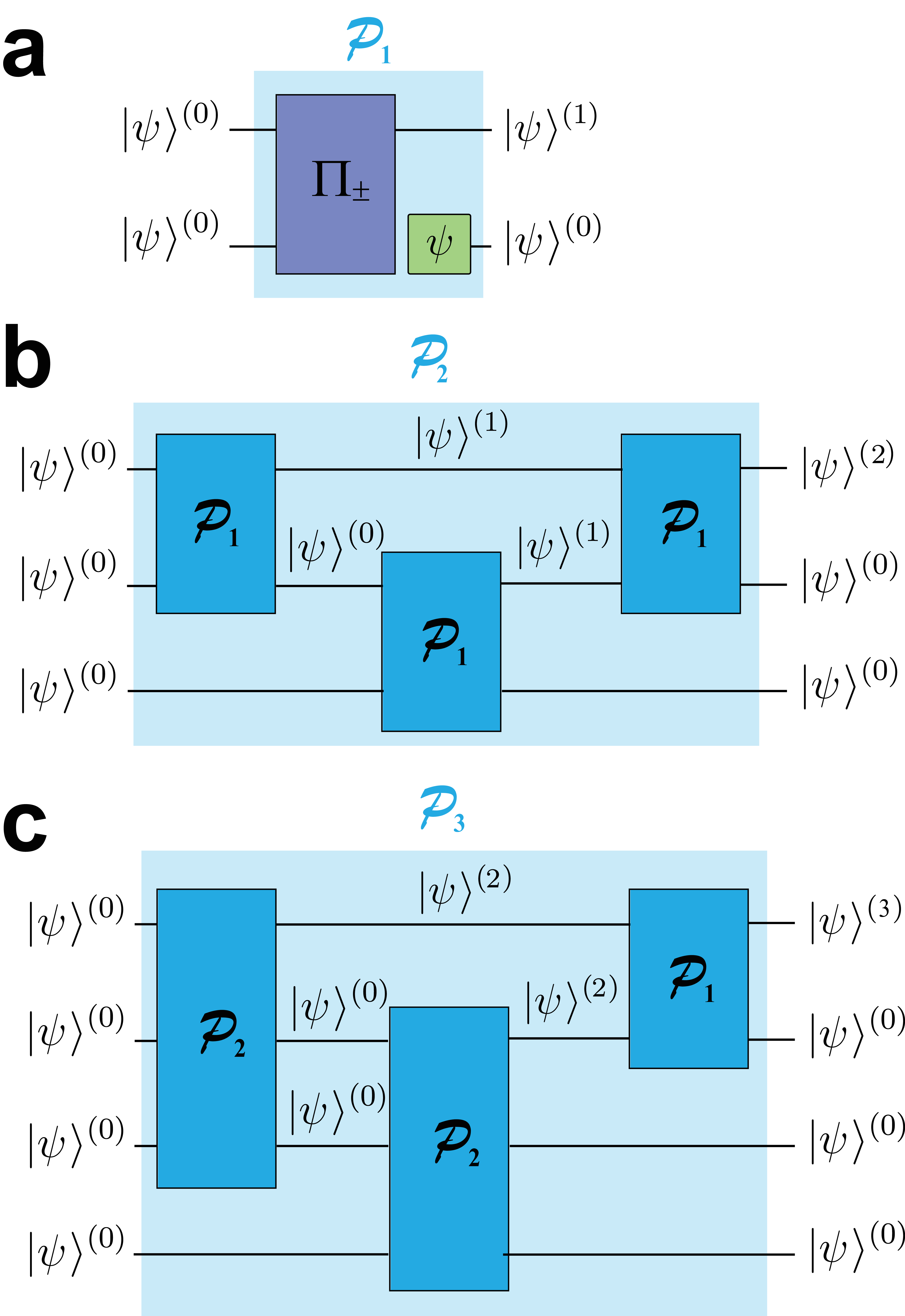}
  \caption{Qubit-recycled implementation of the PQES scheme. Purification sequences for (a) $\ell=1$; (b) $\ell=2$; (c) $\ell=3$ are shown. The block ${\cal P}_\ell$ denotes a circuit that consumes $N=2^\ell$ noisy preparations of $\ket{\psi}^{(0)}$ and outputs one physical branch state together with the SWAP-outcome record required for signed averaging. The signed average over many such runs realizes observables of the purified state. In (a), the two registers are first processed by the SWAP gadget of Fig.~\ref{diag:main_use_case}(b). The discarded register is then reinitialized to a fresh $\ell=0$ input state $\ket{\psi}^{(0)}$.}
  \label{diag:eff_qubits}
\end{figure}

\begin{figure*}[t]
  \centering
  \includegraphics[width=\linewidth]{diagrams/fig3_efficient_qubits_interwoven.eps}
  \caption{Constructing interleaved PQES circuits by pushing unitary blocks backwards through the recycled purification circuit. (a) A common unitary on the output of a $\ell=1$ purification step can be moved to identical unitaries on its two inputs. (b) A unitary on a reset/reinitialization branch can be absorbed into the preparation of the fresh copy. (c) Applying this rule to a $\ell=2$ recycled circuit gives an interleaved circuit for one unitary block. (d) Repeating the same backwards construction gives the $\ell=3$ circuit with two interleaved unitary blocks.}
  \label{diag:eff_qubits_interwoven}
\end{figure*}

\section{Reducing qubit resources by recycling}
\label{sec:protocol}

\subsection{Qubit recycling}
The fully parallel binary-tree PQES circuit, such as shown in Fig.~\ref{diag:main_use_case}(a), uses $N=2^\ell$ noisy copies of the state. This gives a simple depth-efficient binary tree, but it requires an exponentially large number of data registers when the purification power $N$ is increased. In this section, we show how the same SWAP-test tree can be generated sequentially by recycling registers. The number of coherent data qubits is then reduced from $O(M 2^\ell)$ to $O(M\ell)$ for an $M$-qubit register, at the cost of increasing the sequential SWAP-test depth.

In this binary-tree implementation, all $N=2^\ell$ noisy instances of the relevant $M$-qubit register are prepared simultaneously. This implementation therefore uses
\begin{align}
    M\,2^\ell = MN = O(MN)
    \label{eq:treescaling}
\end{align}
data qubits, plus one SWAP-test ancilla for each internal node if all tests in a layer are performed in parallel. Since the tree has $N-1$ internal nodes, the total number of register-level SWAP tests is $N-1$. However, these tests are arranged into only $\ell$ binary-tree layers, so the SWAP-test depth is
\begin{align}
d_{\mathrm{bt}} = O(\ell\, d_{\mathrm{SWAP}}),
\end{align}
where $d_{\mathrm{SWAP}}$ is the depth required to implement one register-level SWAP test on two $M$-qubit registers. The value of $d_{\mathrm{SWAP}}$ is architecture dependent: it is $O(1)$ if the $M$ Fredkin gates within a register-SWAP test can be parallelized with suitable control distribution, and it can scale as $O(M)$ in a serial controlled-SWAP implementation. Thus, Fig.~\ref{diag:main_use_case}(a) should be interpreted as prioritizing depth over width for the signed-outcome SWAP estimator.

A more qubit-efficient implementation is shown in Fig.~\ref{diag:eff_qubits}. In this implementation, registers are recycled so that the full binary tree is generated sequentially rather than stored simultaneously. Here, ${\cal P}_\ell$ denotes a level-$\ell$ implementation that consumes $N=2^\ell$ noisy preparations and returns one branch state together with the complete SWAP-outcome record needed for signed averaging. The signed average over the output records realizes the state ${\cal P}_\ell(\rho)$.

The recycled construction is recursive. A level-$\ell$ instance ${\cal P}_\ell$ first uses ${\cal P}_{\ell-1}$ to generate one level-$(\ell-1)$ branch output, stores that register, then reuses the remaining registers to generate a second level-$(\ell-1)$ branch output, and finally applies one SWAP gadget to merge the two. By induction, this requires only $\ell+1$ live data registers. Hence the coherent data footprint is
\begin{align}
(\ell+1)M = O(M\log_2 N),
\label{eq:recycledwidth}
\end{align}
rather than the $O(MN)$ footprint of (\ref{eq:treescaling}). The price is depth. Let $G_\ell$ denote the number of register-level SWAP-test gadgets executed in a recycled level-$\ell$ implementation. Since the two calls to ${\cal P}_{\ell-1}$ are performed sequentially, the register-level SWAP-test count obeys
\begin{align}
G_\ell = 2G_{\ell-1}+1,\qquad G_0=0,
\end{align}
and therefore
\begin{align}
G_\ell=2^\ell-1=N-1.
\label{eq:swapcount}
\end{align}
Thus, the sequential SWAP-test depth of the fully recycled construction scales as
\begin{align}
d_{\mathrm{recycled}} = O(N\,d_{\mathrm{SWAP}}),
\label{eq:recycleddepth}
\end{align}
up to the depth required to prepare each noisy input copy. The same binary tree of SWAP merges is realized in both implementations; the difference is whether the tree is stored in space or generated in time.

Intermediate realizations are also possible. One may keep several subtrees live in parallel to reduce depth while still using fewer than $MN$ data qubits, or one may recycle aggressively to minimize coherent memory at the cost of longer PQES segments. This width--depth tradeoff is useful because the best implementation depends on the hardware: memory errors, reset speed, connectivity, controlled-SWAP fidelity, and the cost of preparing fresh noisy copies may all determine the optimal schedule.


The resource count above is the cost of the particular architecture studied here: a recursive circuit built from pairwise SWAP tests, recycled registers, and an explicit SWAP-outcome record. This should not be confused with the minimum possible width for estimating traces of powers of a density matrix. If the only task is to estimate quantities such as $\Tr(O\rho^N)$ or $\Tr(\rho^N)$ at the end of a circuit, specialized trace-estimation protocols can be more space- or depth-efficient. For example, reset-based constructions can reduce the coherent width of virtual-distillation estimators to $2M+1$ qubits~\cite{czarnik_qubit_eff}, and multivariate trace-estimation circuits can estimate cyclic-shift traces in constant quantum depth using highly parallelized controlled permutations and GHZ-type control resources~\cite{Quek_2024_const_depth}.

\subsection{Interleaving quantum algorithms}
We now explain how to construct the interleaved circuits shown in Fig.~\ref{diag:eff_qubits_interwoven}(a)(b). The basic object is the normalized SWAP-gadget map from (\ref{eq:generalswapgadget}). For two input states $X$ and $Y$, the resultant state, $\rho_{\sigma}(X,Y)$, and the probability of that state, $P_{\sigma|X,Y}$, are given by (\ref{rhorecursive}) and (\ref{eq:P_sigma_given_ab}), respectively. The SWAP projectors commute with a common unitary applied to both registers,
\begin{align}
    [ \Pi_\pm, U \otimes U ] = 0 .
    \label{eq:bilateralcommute}
\end{align}
Equivalently, the branch maps satisfy
\begin{align}
\rho_{\sigma}(UXU^\dagger, UYU^\dagger) = U \rho_{\sigma}(X,Y) U^\dagger .
\label{eq:branchcovariance}
\end{align}
This identity is shown in Fig.~\ref{diag:eff_qubits_interwoven}(a): a common unitary can be moved through a SWAP purification step and applied to the surviving output register. In the $\ell=1$ purification sequence shown in Fig.~\ref{diag:eff_qubits}(a), one of the two registers is discarded after the SWAP gadget, so only one copy of $U$ appears on the output side of the ${\cal P}_1$ box. The identity in Fig.~\ref{diag:eff_qubits_interwoven}(b) expresses the complementary freedom to absorb a unitary into the preparation of a fresh copy: instead of preparing the lower register in the initial state $|\psi\rangle^{(0)}$, one may prepare the already-updated state $U|\psi\rangle^{(0)}$.


The interleaved circuit should be constructed backwards. Start from the recycled purification circuit with the desired algorithmic unitary block placed at the output. Then move the unitary backwards through the circuit using (\ref{eq:branchcovariance}). Whenever the unitary crosses a SWAP gadget, it becomes the same unitary applied to both input registers of that gadget. Whenever it lands on a register that is freshly prepared after a reset, it is absorbed into the preparation of that fresh copy, as shown in Fig.~\ref{diag:eff_qubits_interwoven}(b). Repeating this procedure gives the interleaved circuits in Fig.~\ref{diag:eff_qubits_interwoven}(c)(d).

The depth-efficient binary-tree circuit in Fig.~\ref{diag:main_use_case}(a) is obtained by the same rule. One may start from a tree in which the unitary block is placed after purification, then commute the unitary backwards through each SWAP layer. Each time it crosses a SWAP test, it splits into identical unitary blocks on the two input copies. 

This backwards construction is important in the recycled implementation. At intermediate times, the live registers can sit at different levels of the recursion. The circuit must therefore preserve the rule that a level-$(k+1)$ purification step combines two corresponding level-$k$ outputs. Pushing the unitaries backwards through the already-defined purification tree guarantees that the same recursive structure is maintained.



There are two limiting schedules. In an output-state purification schedule, one first prepares multiple noisy copies of the final circuit output and then applies the SWAP-test purification tree. In an interleaved schedule, purification layers are inserted between circuit blocks so that errors are suppressed before the full circuit has been completed. For ideal unitaries and ideal SWAP tests, the backwards construction above gives equivalent power-state observables. For realistic circuits, however, the placement of the SWAP layers matters because different schedules expose different intermediate states to noise. In that case, interleaving may alter both the bias and the variance of the final signed estimator. Determining whether intermediate placement gives a genuine error-suppression advantage requires a circuit-level noise model that includes the preparation of corresponding intermediate copies, memory noise, SWAP-test errors, and reset errors. 




\section{Error-model independent performance}


We now analyze the performance of PQES independently of any particular error model. This section contains four general ingredients used in the noise-model calculations below: the single-qubit Bloch-vector update, the spectral action and fixed points of the purification map, general fidelity bounds, and the definition of the error threshold for PQES. These results clarify as to when increasing the number of SWAP-test purification rounds improves the target-state fidelity.

\subsection{Noisy single qubit}
\label{sec:noisyqubit}

As a first example of PQES in action, it is illustrative to consider a single noisy qubit. An arbitrary noisy qubit state can be written
\begin{align}
\rho=\tfrac{1}{2}(I+\vec r\cdot\vec\sigma) .
\label{eq:noisysinglequbit}
\end{align}
Let $r=|\vec r|$ and, for $r\neq 0$, define $\hat r=\vec r/r$. For one purification round, substituting (\ref{eq:noisysinglequbit}) into (\ref{eq:purifiedstaten2}) gives (see Appendix~\ref{app:paulibasis})
\begin{equation}
\label{eq:bloch-map}
    {\cal P} (\rho)
    =\frac{1}{2}\left(I+\frac{2}{1+r^2}\,\vec r\cdot\vec\sigma\right) .
\end{equation}
Equivalently, the Bloch vector is updated as
\begin{align}
\label{eq:qubitradialupdate}
\vec r \rightarrow \frac{2}{1+r^2}\,\vec r .
\end{align}

The multi-round version follows directly from the spectral form of (\ref{eq:purifiedstate}). The eigenvalues of $\rho$ are $(1+r)/2$ and $(1-r)/2$, with eigenvectors aligned and anti-aligned with $\hat r$. For $N=2^\ell$,
\begin{align}
{\cal P}_{\ell}(\rho) = \frac{1}{2}\left(I+r_N\,\hat r\cdot\vec\sigma\right),
\label{multiroundblochstate}
\end{align}
where
\begin{align}
r_N = \frac{(1+r)^N-(1-r)^N} {(1+r)^N+(1-r)^N}.
\label{multiroundblochradius}
\end{align}
For $N=2$, (\ref{multiroundblochradius}) reduces to (\ref{eq:qubitradialupdate}). Thus, for qubits, PQES implements a radial rescaling in the Bloch sphere. The Bloch direction is preserved and the radius monotonically increases whenever $0<|\vec r|<1$, with fixed points at $|\vec r|=0$ and $|\vec r|=1$. From this behavior, one expects PQES to be particularly effective against depolarizing channels: depolarizing noise contracts the Bloch vector without changing its direction, while purification reverses this contraction by repolarizing along the same axis.

\subsection{Spectral purification and fixed points}
\label{sec:purification}

In general, the purification step in (\ref{eq:purifiedstaten2}) acts only on the spectrum of a state. If
\begin{align}
\rho = \sum_i \lambda_i |\lambda_i\rangle\langle\lambda_i| ,
\end{align}
then
\begin{align}
{\cal P}(\rho) = \frac{\rho^2}{\Tr(\rho^2)} = \sum_i \frac{\lambda_i^2} {\sum_j\lambda_j^2} |\lambda_i\rangle\langle\lambda_i| .
\end{align}
The eigenbasis is unchanged, the ordering of eigenvalues is preserved, and the relative contrast between unequal eigenvalues increases. For $\lambda_i>\lambda_j$, the purified eigenvalues obey
\begin{equation}
\label{eq:ratio_increase}
    \frac{\lambda_i'}{\lambda_j'} = \left(\frac{\lambda_i}{\lambda_j}\right)^2 > \frac{\lambda_i}{\lambda_j} .
\end{equation}
Each purification round subsequently pushes the probability weight toward larger eigenvalues and away from smaller ones. Iterating the map drives the state toward the dominant eigenspace, as described by (\ref{spectralconvergencebound}).


The map also increases the purity, except when the nonzero eigenvalues are already all equal. As shown in Appendix~\ref{app:purification},
\begin{align}
\Tr[{\cal P}(\rho)^2]\ge \Tr(\rho^2),
\label{eq:purity_increase}
\end{align}
with equality if and only if $\rho$ is maximally mixed on its support. This includes pure states: when the support has rank one, the state is already fixed and has purity one. The full-rank maximally mixed state is the opposite extreme. More generally, the fixed points are precisely the states that are uniform on their support.

\subsection{Spectral convergence and the required number of purification rounds}

The convergence of ${\cal P}_{\ell}(\rho)$ is controlled by the spectrum of $\rho$. Let
\begin{align}
\rho=\sum_{i=1}^{D}\lambda_i |\lambda_i\rangle\langle \lambda_i|,
\qquad
\lambda_1>\lambda_2\ge \lambda_3\ge \cdots \ge 0,
\end{align}
where the largest eigenvalue is assumed to be nondegenerate. Setting $N=2^\ell$, (\ref{eq:purifiedstate}) can be rewritten as
\begin{align}
{\cal P}_{\ell}(\rho) = \frac{ |\lambda_1\rangle\langle \lambda_1| + \sum_{i>1}\left(\lambda_i/\lambda_1\right)^N |\lambda_i\rangle\langle \lambda_i|
}{ 1+ \sum_{i>1}\left(\lambda_i/\lambda_1\right)^N }.
\label{spectralpell}
\end{align}
Thus ${\cal P}_{\ell}(\rho)$ converges exponentially in $N$ to the dominant eigenvector $|\lambda_1\rangle$ whenever $\lambda_1>\lambda_2$. Let
\begin{align}
r_\lambda=\frac{\lambda_2}{\lambda_1},
\qquad
R_N=\sum_{i>1}\left(\frac{\lambda_i}{\lambda_1}\right)^N .
\end{align}
Since ${\cal P}_{\ell}(\rho)$ and $|\lambda_1\rangle\langle \lambda_1|$ commute, their trace distance is simply the total weight outside the dominant eigenvector:
\begin{align}
T\!\left({\cal P}_{\ell}(\rho), |\lambda_1 \rangle\langle \lambda_1|\right)
&=
\frac{1}{2} \left\| {\cal P}_{\ell}(\rho)-|\lambda_1 \rangle\langle \lambda_1| \right\|_1 \nonumber\\
&= \frac{R_N}{1+R_N}.
\label{traceexactrN}
\end{align}
Furthermore,
\begin{align}
R_N = \sum_{i>1}\left(\frac{\lambda_i}{\lambda_1}\right)^N \le (D-1)r_{\lambda}^N,
\end{align}
and therefore
\begin{align}
T\!\left({\cal P}_{\ell}(\rho), |\lambda_1 \rangle\langle \lambda_1|\right) \le (D-1) \left(\frac{\lambda_2}{\lambda_1}\right)^{2^\ell}.
\label{spectralconvergencebound}
\end{align}
To guarantee that the finite-round spectral error is at most $\eta$, it is sufficient that
\begin{align}
(D-1)
\left(\frac{\lambda_2}{\lambda_1}\right)^{2^\ell} \le \eta .
\end{align}
Equivalently, since $\lambda_1/\lambda_2>1$,
\begin{align}
2^\ell \ge \frac{\log[(D-1)/\eta]}{\log(\lambda_1/\lambda_2)} .
\end{align}
A sufficient number of purification rounds is then,
\begin{align}
\ell \ge \left\lceil \log_2 \left[ \frac{\log[(D-1)/\eta]} {\log(\lambda_1/\lambda_2)} \right] \right\rceil .
\label{ellbound}
\end{align}
This bound separates the algebraic depth $\ell$ of the recursive SWAP tree from the spectral gap of the noisy state. If the dominant eigenvalue is degenerate, the same argument shows convergence to the normalized projector onto the dominant eigenspace rather than to a unique pure state. In applications, the dominant eigenvector itself may differ from the ideal noiseless state.

\subsection{Fidelity bounds and dominant-eigenvector limitation}

Consider $\rho$, a noisy version of $|\psi\rangle\!\langle\psi|$ produced by a physical noise map prior to purification. For $N=2^\ell$, the fidelity after $\ell$ rounds of PQES is
\begin{equation}
\label{eq:Fq-master}
 F_N := \bra{\psi} {{\cal P}}_\ell (\rho) \ket{\psi} = \frac{\bra{\psi}\rho^N\ket{\psi}}{\Tr(\rho^N)} .
\end{equation}
%

We can obtain general bounds on $F_N$ by analyzing $\bra{\psi}\rho^N\ket{\psi}$. Let
\begin{align}
F := \bra{\psi}\rho\ket{\psi}
\end{align}
be the original fidelity. Writing $\rho=\sum_i\lambda_i |\lambda_i\rangle\langle  \lambda_i|$ and $p_i=|\langle \lambda_i|\psi\rangle|^2$, we have
\begin{align}
\bra{\psi}\rho^N\ket{\psi} = \sum_i p_i\lambda_i^N .
\end{align}
Since $x^N$ is convex on $[0,1]$ for $N\ge1$, Jensen's inequality gives
\begin{align}
\bra{\psi}\rho^N\ket{\psi} = \sum_i p_i\lambda_i^N \ge \left(\sum_i p_i\lambda_i\right)^N = F^N .
\end{align}
On the other hand, $0\le\rho\le I$ implies $\rho^N\le \rho$ in the operator order, and hence
\begin{align}
\bra{\psi}\rho^N\ket{\psi}\le \bra{\psi}\rho\ket{\psi}=F .
\end{align}
Therefore, for every input state $\rho$,
\begin{equation}
\label{eq:Fq-sandwich}
    \frac{F^N}{\Tr(\rho^N)} \le  F_N  \le  \frac{F}{\Tr(\rho^N)} .
\end{equation}
%
%
Specific noise models and state families sharpen this picture by expressing $F$, $\Tr(\rho^N)$, and $\bra{\psi}\rho^N\ket{\psi}$ as explicit functions of the noise parameters.

These bounds should be interpreted together with the dominant-eigenvector limitation of the power map. As $N$ increases, ${\cal P}_\ell(\rho)$ converges to the dominant eigenspace of $\rho$, not necessarily to the ideal noiseless state. If the dominant eigenvector of $\rho$ is the target $|\psi\rangle$, then increasing $\ell$ can drive $F_N$ to unity. If, instead, the dominant eigenvector is shifted away from $|\psi\rangle$, the fidelity saturates at a value determined by this mismatch. This is the coherent-mismatch or dominant-eigenvector floor discussed in virtual distillation and error suppression analyses \cite{Huggins2021,Koczor_derangement,Koczor2021_dominantEV}.

\subsection{Error threshold}
\label{subsec:threshold-performance}

In conventional quantum error correction, an error threshold separates the regime in which increasing the code size suppresses the logical error rate from the regime in which additional encoding no longer improves the logical state. Below such a threshold, the logical error rate may be made arbitrarily small by increasing the redundancy of the code.

PQES has an analogous scalable redundancy parameter. Rather than increasing a code distance, PQES increases the number of noisy copies, $N=2^\ell$, or equivalently, the purification depth $\ell$. We therefore define an error threshold for PQES as the largest physical error rate below which increasing $\ell$ suppresses the effective logical error rate to zero. Throughout the present threshold analysis, the SWAP tests, measurements, resets, and classical processing are treated as ideal. The resulting quantity is therefore analogous to a code-capacity threshold rather than a full circuit-level threshold.




\begin{definition}[Error threshold for PQES]
For a given noise model, we define the PQES error threshold $p_{\mathrm{th}}$ as the largest physical error rate for which the target remains asymptotically recoverable under ideal PQES cycles, i.e., for all $p < p_{\mathrm{th}}$, the effective logical error rate with PQES satisfies
\begin{equation}
    \forall\,p < p_{\mathrm{th}}:\quad \lim_{\ell \to \infty} \gamma_{L} (\ell,p) = 0 ,
    \label{errorthresholddef}
\end{equation}
where $\gamma_{L}(\ell,p)$ is the effective logical error rate extracted from the fidelity decay, and $\ell$ denotes the number of purification rounds per PQES cycle.
\end{definition}

We use the term ``logical error rate'' here as an effective repeated-cycle diagnostic for the decay of the target-state fidelity under alternating noise and ideal PQES layers. To determine the threshold in practice, we follow the procedure as follows:
\begin{enumerate}
    \item[0)] Initialize the state in a suitable pure state $\rho = | \psi_0 \rangle\langle \psi_0 |$. Set the cycle number $t = 0$.
   \item[1)] Apply the error channel
   \begin{align}
   \rho \rightarrow {\cal E} (\rho) .
   \end{align}
 \item[2)] Perform $\ell$ rounds of purification
   \begin{align}
   \rho \rightarrow {\cal P}_{\ell} (\rho) .
   \end{align}
 \item[3)] Measure the fidelity $F= \langle \psi_0 | \rho | \psi_0 \rangle$.
  \item[4)] Update $t \rightarrow t +1$ and go to step 1.
\end{enumerate}
To extract $\gamma_{L}$, we assume that the fidelity $F$ follows an exponential decay with the number of cycles and evaluate
\begin{align}
    \gamma_{L} =  -\left.\frac{dF}{dt} \right|_{t=0}  \approx
    F(t=0) - F(t=1),
    \label{gammadef}
\end{align}
where $t$ is the number of cycles. This gives the initial decay rate of the fidelity after applying $\ell$ rounds of PQES in each cycle. The threshold $p_{\mathrm{th}}$ is then obtained by applying (\ref{errorthresholddef}): below threshold, increasing the purification depth suppresses the effective decay rate toward zero, while above threshold the noisy state is no longer spectrally recoverable by the power map alone.


\section{Global depolarizing errors}
\label{sec:global_depol}

As seen in Sec.~\ref{sec:noisyqubit}, PQES is especially natural for depolarizing noise because the purification map sharpens the spectrum without changing the eigenbasis. A global depolarizing channel is therefore the cleanest setting in which to illustrate the basic recoverability mechanism: the target state remains an eigenvector of the noisy state, and purification amplifies its eigenvalue relative to the uniformly distributed error components. We now analyze this case explicitly.

\subsection{Error channel}

The global depolarizing channel is defined as
\begin{equation}
\label{eq:isotropic-family}
   {\cal E} ( \rho ) = (1-p) \rho + p \frac{I}{D}  ,
\end{equation}
where $ p \in  [0,1] $ is the error probability. Note that we consider the channel (\ref{eq:isotropic-family}) to be applied {\it per quantum register}.  For example, in Fig. \ref{diag:main_use_case},  the quantum channel acts on each copy of $ | \psi_n \rangle^{(\ell)} $.

\begin{figure}[t]
  \centering
  \includegraphics[width=0.78\linewidth]{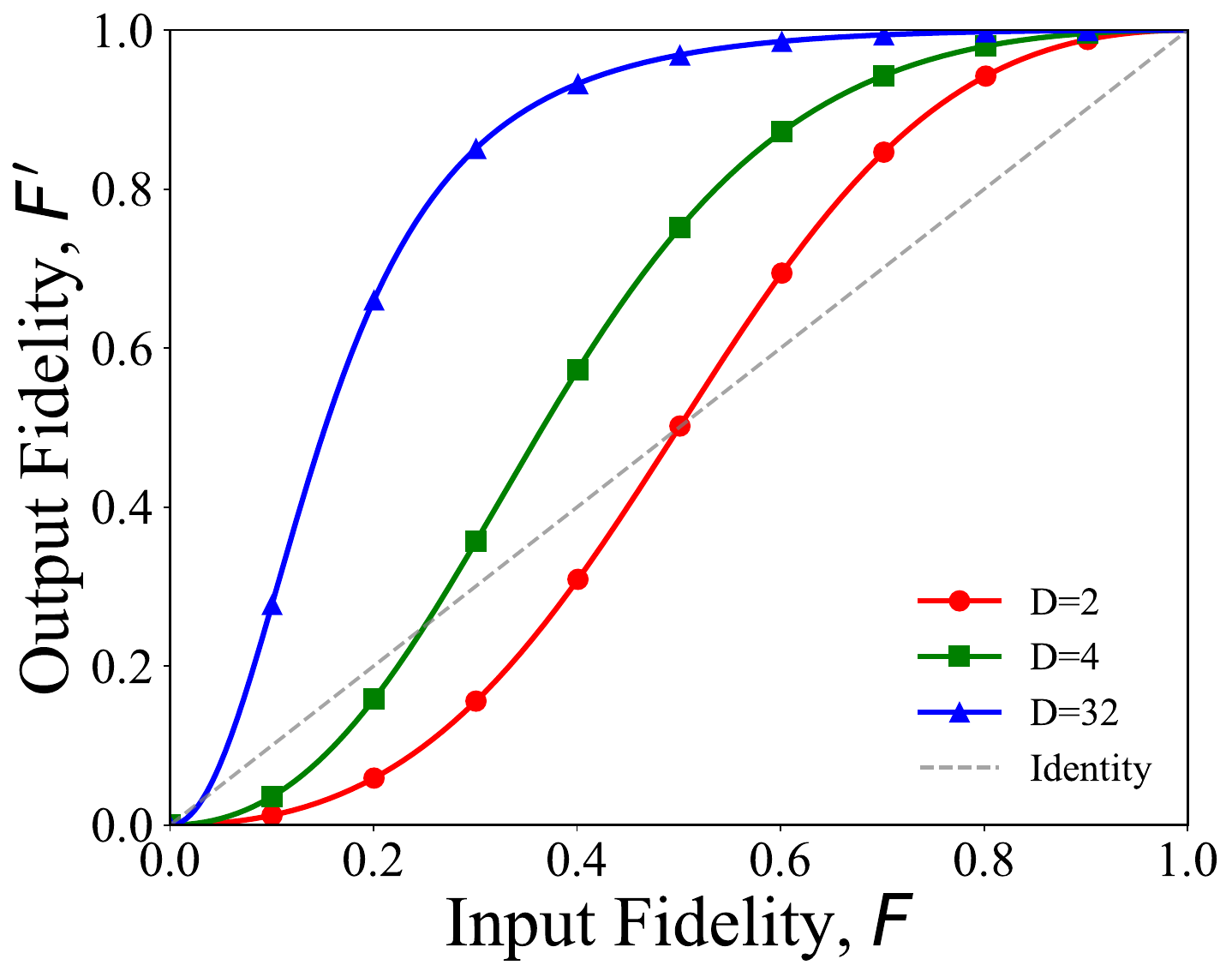}
  \caption{%
   Output fidelity versus input fidelity after a single purification round for the global depolarization channel \eqref{eq:isotropic-family}, for several register dimensions $D$. Curves show the exact map of \eqref{eq:gen_error_reduction} with the gray dashed line indicating the identity $F' = F$. The purification step strictly improves fidelity when $F> \frac{1}{D}$, and the gain is largest at lower $F$ where the curvature is steepest.  }
  \label{fig:fout_vs_f_isotropic}
\end{figure}

\subsection{Purification}

We first examine how the purification operation (\ref{eq:purifiedstate}) improves the fidelity after the channel (\ref{eq:isotropic-family}) acts on a initial state  $ | \psi_0 \rangle \langle \psi_0 | $.  The global depolarizing channel is a particularly simple case to analyze because the channel leaves it a Werner state, i.e. isotropic, at all times
\begin{align}
\label{wernerrho}
\rho = (1- \lambda) | \psi_0 \rangle \langle \psi_0 | + \lambda \frac{I}{D} .
\end{align}
This family is closed under PQES because the state has only two distinct eigenspaces: the target direction $|\psi_0\rangle$ and the orthogonal $(D-1)$-dimensional subspace. Powers of $\rho$ preserve this eigenspace decomposition. The fidelity of the Werner state (\ref{wernerrho}) is 
\begin{align}
F = \langle \psi_0 | \rho | \psi_0 \rangle =   1-\lambda(1 - \frac{1}{D} ) .
\label{fidpdepol}
\end{align}
Substituting (\ref{wernerrho}) into (\ref{eq:generalswapgadget}), we obtain the effect of applying one round of purification

\begin{align}
    \lambda \mapsto \lambda'= \frac{\lambda^2}{D(1-\lambda)^2-\lambda(\lambda-2)} .
    \label{PQESupdateglobal}
\end{align}
This can be equally written in terms of the fidelity using the relation (\ref{fidpdepol}), giving

\begin{align}
\label{eq:gen_error_reduction}
 F \mapsto   F' = \frac{ F^2}{F^2+\frac{(1-F)^2}{D-1}}  .
\end{align}
For a general purification power $N=2^\ell$, the fidelity after PQES is given by
\begin{align}
 F \mapsto F_N = \frac{F^N}{F^N+\dfrac{(1-F)^N}{(D-1)^{N-1}}}.
\label{eq:global_fN}
\end{align}

In Fig. \ref{fig:fout_vs_f_isotropic}, we show the output fidelity for various system sizes. By solving for the crossing point $ F ' = F $, we find that the fidelity improves in the region $ F > 1/D $.  In terms of the Werner mixing parameter, this corresponds to $ \lambda < 1 $.  For large registers $D \gg 1 $, the purification becomes extremely effective and quickly approaches $ F' \approx 1 $ after a single round.  We show in Appendix \ref{app:fidelity_convergence} that if $ F > 1/D $, iterating (\ref{eq:gen_error_reduction}) is strictly increasing and $ F_* = 1 $ is an attractive fixed point. 

Thus, within the Werner family and assuming ideal PQES layers, any state with $\lambda<1$ remains spectrally recoverable: the target eigenvalue is strictly larger than the orthogonal eigenvalues, and the limit $N\to\infty$ gives unit fidelity. At $\lambda=1$, the state is completely mixed and the target information has been erased.

\subsection{PQES error threshold}
\label{subsubsec:single-qubit-example}

\begin{figure}[t]
  \centering
    \includegraphics[width=\linewidth]{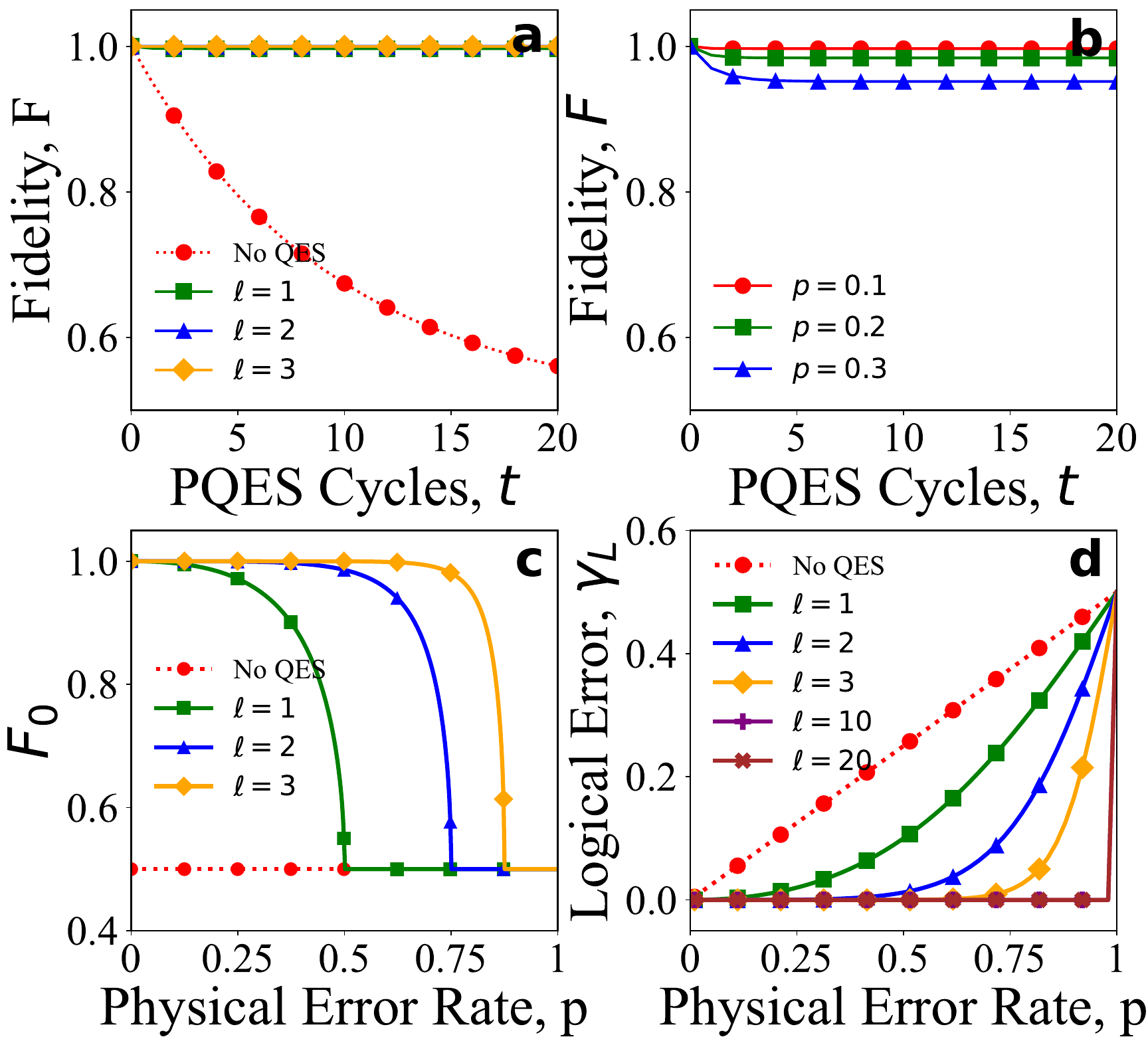}
    \caption{PQES error threshold behavior for a single qubit ($D=2$) in a global depolarizing channel with PQES. (a)(b)  Fidelity evolution under the depolarizing Bloch vector contraction (\ref{eq:isotropic-family}) and $ \ell $ rounds of purification (\ref{eq:qubitradialupdate}), where $ t $ is the number of cycles.  (a) Fidelity evolution versus $ t$ with $ p = 0.1 $ and various  $ \ell $. (b) Fidelity evolution versus $ t$ for $ \ell = 1 $ and various $ p $ as shown.  (c) The steady state ($ t \rightarrow \infty $) fidelity $ F_0 $ for various $ \ell $.  (d) The effective logical error rate $ \gamma_{L} $ evaluated using (\ref{gammadef})  for various $ \ell $. }
    \label{fig:onequbitrec}
\end{figure}

To determine the error threshold under the global depolarizing channel, we follow the procedure given in Sec.~\ref{subsec:threshold-performance}. Under the global depolarizing channel, the Werner state parameter undergoes the update rule
\begin{align}
\label{globaldepolmap}
    \lambda \mapsto (1-p)\lambda + p  .
\end{align}
PQES then updates the Werner parameter according to (\ref{PQESupdateglobal}). The fidelity is then evaluated using (\ref{fidpdepol}).

A typical sequence is shown in Fig. \ref{fig:onequbitrec}(a)(b). We see that the fidelity follows an exponential decay due to the depolarizing channel being applied per cycle.  Without any application of PQES ($ \ell = 0 $), the fidelity approaches the completely mixed state.  As expected, adding additional rounds of purification reduces the error rate.  Interestingly, the fidelities converge to higher values in the limit of a large number of cycles, improving with $ \ell $. For example, in the case of $ \ell =1 $, the steady state fidelity may be found by determining the fixed point of the combined effect of (\ref{PQESupdateglobal}) and (\ref{globaldepolmap}), giving
\begin{align}
    F_0 & := \lim_{t\rightarrow \infty} F(t) \nonumber \\
    & = \frac{1}{2}\left(
    1 + \sqrt{1-\frac{4(D-1)p^2}{D^2(1-p)^2}}
    \right).
\end{align}
For a single qubit, this expression is real for $p\le 1/2$, consistent with the finite-$\ell$ behavior in Fig.~\ref{fig:onequbitrec}. Larger values of $\ell$ extend the range over which the high-fidelity fixed point is maintained, approaching the asymptotic error threshold $p_{\mathrm{th}}=1$ as $\ell\rightarrow\infty$. The steady-state fidelities as a function of $p$ are shown in Fig.~\ref{fig:onequbitrec}(c). This saturation effect preserves fidelity in addition to suppressing the effective logical decay rate.

Figure \ref{fig:onequbitrec}(d) shows the effective logical error rate $ \gamma_L $ as a function of physical error rate $ p $ for various numbers of rounds $ \ell $ of PQES.  We see that, as expected, PQES reduces the effective logical error rate as $ \ell $ increases.  With $ \ell = 20 $, the effective logical error rate is reduced to nearly zero for any $ p < 1$. The PQES error threshold is found by determining the crossing point of the curves in Fig. \ref{fig:onequbitrec}(d), which in this case is 
\begin{align}
    p_{\text{th}} = 1 ,
    \label{pthglobal}
\end{align}
for the global depolarizing channel. This means that for any $p<1$, the target state remains spectrally recoverable in the ideal-purification limit: increasing $\ell$ suppresses the subdominant eigencomponents of the noisy state. Operationally, this statement assumes ideal SWAP-test layers and sufficient copies and samples to resolve the signed estimator.

The spectral origin of this threshold is especially transparent in this model. After one application of the global depolarizing channel to $|\psi_0\rangle\langle\psi_0|$, the target eigenvalue is
\begin{align}
\lambda_{\mathrm{tar}}=1-p+\frac{p}{D},
\end{align}
while each orthogonal eigenvalue is
\begin{align}
\lambda_{\perp}=\frac{p}{D}.
\end{align}
Hence
\begin{align}
\lambda_{\mathrm{tar}}-\lambda_{\perp}=1-p .
\end{align}
For every $p<1$, the target state remains the unique dominant eigenvector of the noisy density matrix. Therefore, in the limit $\ell\rightarrow\infty$, the power map converges back to $|\psi_0\rangle\langle\psi_0|$. At $p=1$, all eigenvalues are equal and the state is completely mixed, so the target information has been erased. This gives the PQES error threshold (\ref{pthglobal}).

Similar results to Fig. \ref{fig:onequbitrec} are obtained for higher dimensions $ D $. The same spectral argument applies: for the global depolarizing channel, the target eigenvalue remains strictly larger than every orthogonal eigenvalue for all $p<1$, while $p=1$ gives the completely mixed state. Hence, the global depolarizing PQES error threshold is $p_{\mathrm{th}}=1$ for every register dimension.


\section{Local depolarizing errors}
\label{sec:local_depol}

The global depolarizing channel analyzed in Sec.~\ref{sec:global_depol} has the advantage of computational simplicity, due to the state being of Werner form at all times. It is, however, a highly symmetric error model, since it is equivalent to a sum of all possible Pauli errors, which involves multiqubit bit/phase flip errors. In order to examine a more realistic case, we now examine local depolarizing errors. Here, each qubit is subjected to an identically and independently distributed single-qubit depolarizing channel.

\subsection{Error channel}

The local depolarizing channel on the $ m $th qubit is defined as
\begin{align}
\label{eq:Dp-def}
    \mathcal E_m (\rho)&=E_m^{(0)} \rho {E_m^{(0)} }^\dagger+\sum_{j=x,y,z}E_m^{(j)} \rho {E_m^{(j)}}^\dagger ,
\end{align}
where the Kraus operators are defined as
\begin{align}
\label{eq:pauli-kraus}
    E_m^{(0)} & =\sqrt{1-p} I_m \nonumber \\
    E_m^{(j)}  & =\sqrt{\tfrac{p}{3}} \sigma_m^{(j)}
\end{align}
where $\sigma_m^{(j)}$ denotes the Pauli operator $\sigma^{(j)}$ acting on qubit $m$ and identity on all other qubits, and $ p \in [0,1] $ is the error probability. The full register noise is then the product channel
%
\begin{equation}
    \mathcal{E}(\rho) =
    \left(\mathcal E_1\circ\mathcal E_2\circ\cdots\circ\mathcal E_M\right)(\rho),
    \label{localdepolchannelall}
\end{equation}

\subsection{Purification}

\subsubsection{Symmetric product states}

\begin{figure}[t]
  \centering
    \includegraphics[width=\linewidth]{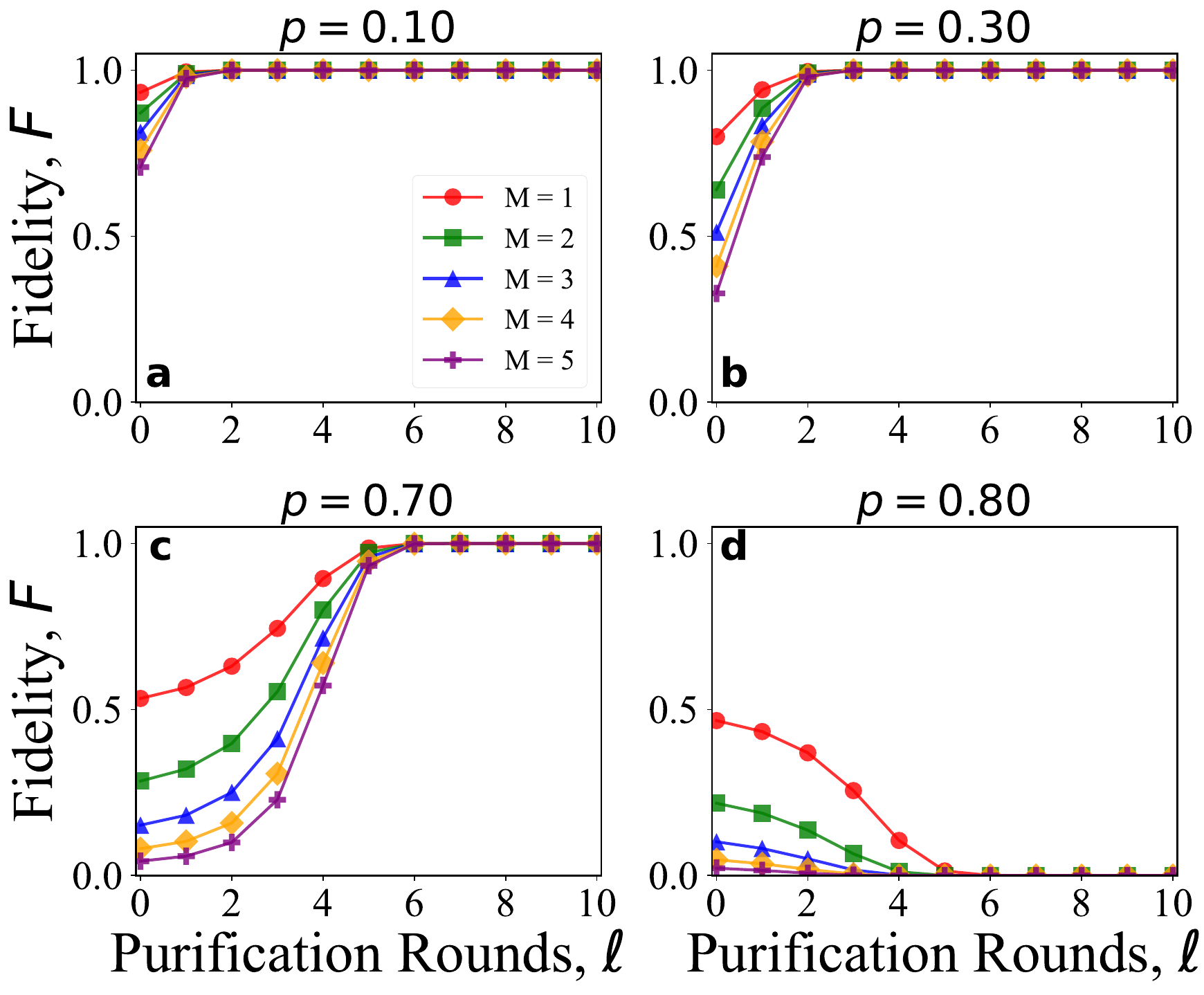}
    \caption{Fidelity improvement of locally depolarized states under successive SWAP purification rounds.  All plots show fidelity, $F$, versus rounds of SWAP purification, $\ell$, for various qubit numbers $ M $.  The initial state is the state $ | \psi_0 \rangle = \ket{+}^{\otimes M}$  subjected to the channel (\ref{eq:Dp-def}) with the error probability (a) $ p = 0.1 $; (b) $ p = 0.3 $; (c) $ p = 0.7 $; (d) $ p = 0.8 $.    
    }
    \label{fig:fidelity_grid_depolarizing}
\end{figure}

We first show the effect of PQES applied to the $M$-qubit symmetric product state, subjected to local depolarization. For concreteness, we choose $|\psi_0\rangle=\ket{+}^{\otimes M}$. Because the single-qubit depolarizing channel is isotropic, the same error threshold applies to any product state with identical single-qubit factors, $|\phi\rangle^{\otimes M}$.


For a single qubit initially in $\ket{+}$, the local depolarizing channel gives
%
\begin{align}
\mathcal E_m(\ket{+}\bra{+}) = \left(1-\frac{4p}{3}\right)\ket{+}\bra{+} + \frac{4p}{3}\frac{I}{2}.
\end{align}
The $M$-qubit noisy state is then
\begin{align}
\mathcal{E}(\ket{+}\bra{+}^{\otimes M} ) = \left[ \left(1-\frac{4p}{3}\right)\ket{+}\bra{+} + \frac{4p}{3}\frac{I}{2} \right]^{\otimes M}.
\label{localdepolstate}
\end{align}

The error threshold follows from the single-qubit eigenvalues. In the $\{\ket{+},\ket{-}\}$ basis, the noisy one-qubit state has eigenvalues
\begin{align}
F_+(p)=1-\frac{2p}{3},
\qquad
F_-(p)=\frac{2p}{3}.
\end{align}
The target eigenvector $\ket{+}$ is dominant if and only if $F_+(p)>F_-(p)$, i.e.
\begin{align}
p<\frac{3}{4}.
\end{align}
For the product state, the eigenvalue of $\ket{+}^{\otimes M}$ is $F_+(p)^M$, while the eigenvalue of a computational basis string in the $X$ basis with $k$ minus signs is $F_+(p)^{M-k}F_-(p)^k$. Hence $\ket{+}^{\otimes M}$ is the unique dominant eigenvector precisely when $p<3/4$. At $p=3/4$, $F_+=F_-=1/2$ and the state is completely mixed. For $p>3/4$, the dominant eigenvector is no longer the target; in the extreme case $p=1$, the one-qubit state has eigenvalues $1/3$ on $\ket{+}$ and $2/3$ on $\ket{-}$, so repeated purification converges toward $\ket{-}^{\otimes M}$ rather than $\ket{+}^{\otimes M}$.

Figure~\ref{fig:fidelity_grid_depolarizing} shows the fidelity after $\ell$ purification rounds and various physical error probabilities $p$. The local depolarization noise model is applied only to the initial state, and successive purification rounds are applied to the noisy state.  We see that, for $ p < 3/4 $, the fidelity converges to 1 for all system sizes.  Larger system sizes tend to require more purification rounds to reach the same fidelity, but this may be attributed to starting at a lower fidelity.  For moderate $p$ (e.g.  $p \lesssim 0.5$), only a few rounds are needed to reach $F_\ell \gtrsim 0.99$.  As $p$ approaches $3/4$ from below, more rounds are required. Once the error probability is larger than $ p \ge 3/4 $, the fidelities  worsen with $ \ell $.


\begin{figure}[t!]
  \centering
    \includegraphics[width=\linewidth]{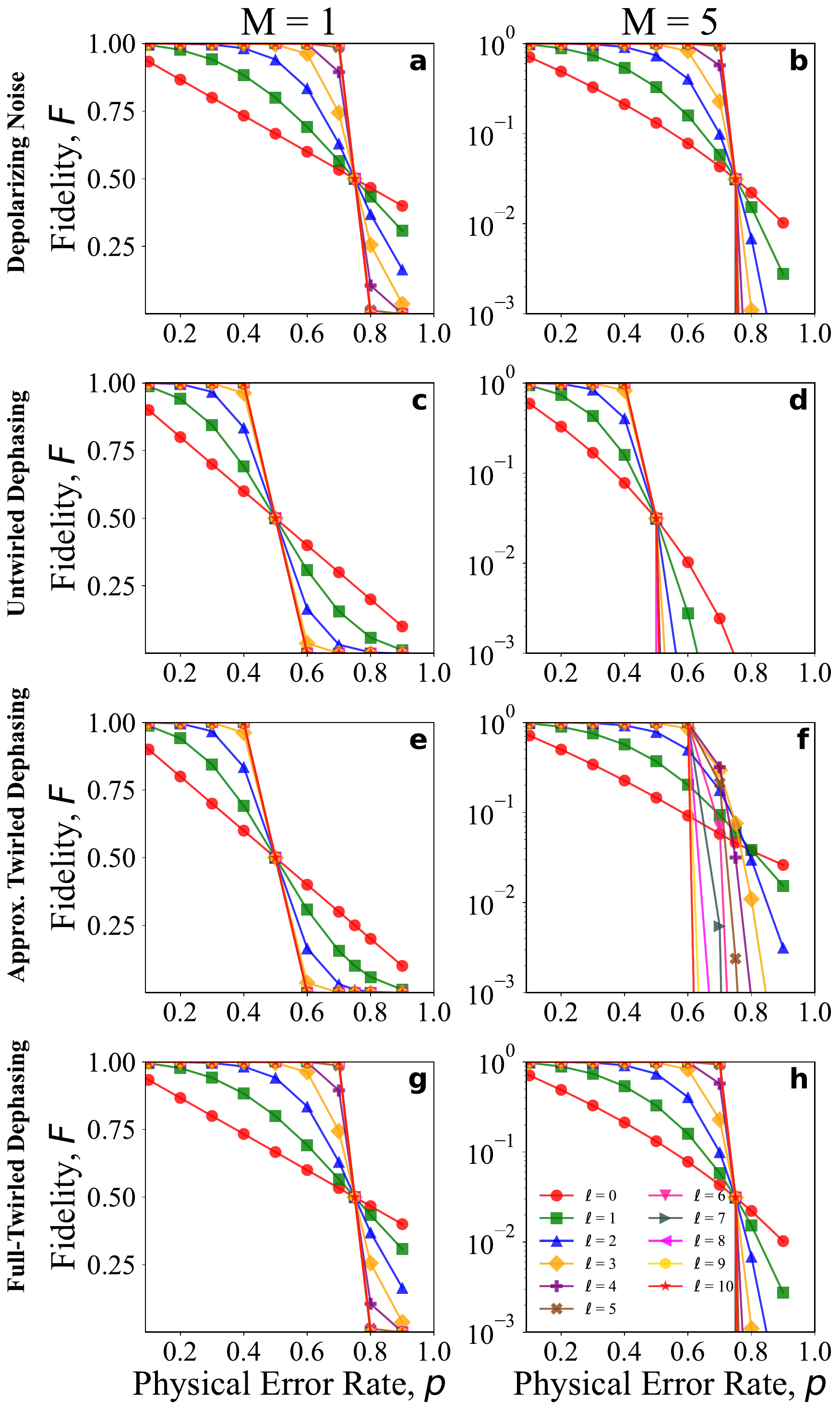}
    \caption{Fidelity evolution of various noise models under one cycle of PQES ($t=1$).  All plots show fidelity, $F$, versus physical error rate, $p$, for various numbers of purification rounds, $\ell$.   The initial state is $ | \psi_0 \rangle = \ket{+}^{\otimes M}$; The columns represent $M=1$ (left) and $M=5$ (right). The top row (a)(b) is subjected to the local depolarization channel (\ref{eq:Dp-def}); the second row (c)(d) is subjected to local dephasing without twirling (\ref{eq:Z_dephasing_channel_main}); the third row (e)(f) is subjected to local dephasing with $20\%$ twirling (\ref{eq:twirl-1qubit}), where the applied twirling gates are randomly selected from the total set of $3^M$ options; the fourth row (g)(h) is subjected to local dephasing with full twirling (\ref{eq:twirl-1qubit}). The protocol exhibits high spectral recoverability for depolarizing noise up to the product-state threshold $p_{\mathrm{th}}=3/4$. The fully-twirled dephasing matches the depolarization result as the channel is effectively isotropic after twirling.}
    \label{fig:SWAP_thresholds}
\end{figure}

We may clearly see the critical value of $ p $ by plotting the fidelities versus $ p $ as shown in Fig. \ref{fig:SWAP_thresholds}(a)(b).  For $p<3/4$, each purification level increases fidelity, and as $\ell$ grows, the fidelity  can be made arbitrarily close to $1$. The critical value remains $ p = 3/4 $ for all values of $ M $, consistent with the spectral analysis.

\subsubsection{Small error expansion}
\label{subsubsec:local-depol}

The results of the previous section were limited to permutationally symmetric product states. In order to see the performance of PQES on more general states, we perform a small error expansion under the local depolarizing channel.

For an arbitrary pure state $ | \psi_0 \rangle $, we may expand the density matrix in the $M$-qubit Pauli basis as
\begin{equation}
\ket{\psi_0}\!\bra{\psi_0} =  \frac{1}{2^M}\sum_{P} r_P P
\label{pauliexp}
\end{equation}
where the sum runs over $4^M $ Pauli strings and the real coefficients $r_P := \bra{\psi} P \ket{\psi} $ satisfy $\sum_P r_P^2 = 2^M$.  Each Pauli string has a weight $ w(P) $, which is equal to the number of non-identity Pauli operators.  The coefficients, $r_P $, may then be grouped by weight, so that we can define
\begin{equation}
    a_k =  \frac{1}{2^M}\sum_{P:\,w(P)=k} r_P^2,
\end{equation}
By performing a small-$ p $ expansion, we evaluate that the fidelity of a general state decays under the local depolarizing channel as (see Appendix \ref{app:smallerror})
\begin{align}
    F(\rho) &= 1 - \frac{4}{3}\,\bar{k}\,p + O(p^2)
\end{align}
where $ \bar{k} := \sum_k k a_k $. Under one round of PQES with purification (\ref{eq:purifiedstaten2}), the output fidelity is
\begin{equation}
    F({\cal P}(\rho))
    =\frac{\bra{\psi_0}\rho^2\ket{\psi_0}}{\Tr(\rho^2)}
    = 1 - O(p^2).
\end{equation}
Thus, for a single purification round, the leading logical error $ O(p) $ is removed entirely.  This shows the effectiveness of the PQES against local depolarizing errors for the general case in the small-$p$ regime.

\subsection{PQES error threshold}

To analyze the PQES error threshold behavior, we again follow the procedure in Sec. \ref{subsec:threshold-performance} using the local depolarizing channel (\ref{localdepolchannelall}) starting from the initial state $ |\psi_0 \rangle = |+\rangle^{\otimes M} $.  

Figure \ref{fig:combined_fidelity_logical_error_depol}(a)(b) shows the fidelity evolution as a function of error and purification cycles for $ M = 5 $.  Similarly to the global depolarizing channel (see Fig. \ref{fig:onequbitrec}), the fidelity follows an exponential decay that saturates to a finite value for the physical error values shown.  For the physical error probabilities in Figs. \ref{fig:combined_fidelity_logical_error_depol}(a)(b), increasing $\ell$ reduces the effective logical error rate (\ref{gammadef}). 

Figure \ref{fig:combined_fidelity_logical_error_depol}(c)(d) shows the effective logical error rate $ \gamma_L $ as a function of the physical error rate $ p $. For both $ M = 1 $ and $ M= 5$, we observe a  crossover of the curves, giving the PQES error threshold at 
\begin{align}
    p_{\mathrm{th}} = 3/4 .
\end{align}
This threshold agrees with the spectral condition derived above: below $p_{\mathrm{th}}$, the target product state remains the unique dominant eigenvector of the locally depolarized density matrix; above $p_{\mathrm{th}}$, the dominant eigenvector changes. For $p<p_{\mathrm{th}}$, increasing the number of purification rounds $\ell$ substantially suppresses the effective logical error rate $\gamma_{L} $, reflecting that each PQES cycle yields a net purification gain that slows logical-error accumulation over time. As $p\to p_{\mathrm{th}}$, the benefit of increasing $\ell$ diminishes, and progressively larger $\ell$ is required to maintain a small effective decay rate. For $p>p_{\mathrm{th}}$, additional purification rounds amplify the wrong dominant eigenvector, so the target fidelity decreases rather than improves. The value $p_{\mathrm{th}}=3/4$ also appears in other nonlinear purification-based error-suppression settings~\cite{grafe2025ultrahigh}, reflecting the same single-qubit spectral crossing.

We note that the $M =1 $ case is physically the same as the global depolarizing channel for a single qubit $ D=2 $  The difference of Fig. \ref{fig:combined_fidelity_logical_error_depol}(c) to Fig. \ref{fig:onequbitrec}(d) arises only due to the difference in definition of the error channel.  For the global depolarizing channel, $p = 1 $ corresponds to a completely mixed state, whereas for the local depolarizing channel, $p = 3/4 $ is the corresponding point.  The $ M=5 $ case, shown in Fig. \ref{fig:combined_fidelity_logical_error_depol}(d), is distinct from the global depolarizing channel, since only local depolarization is applied.

\begin{figure}[t]
  \centering
    \includegraphics[width=\linewidth]{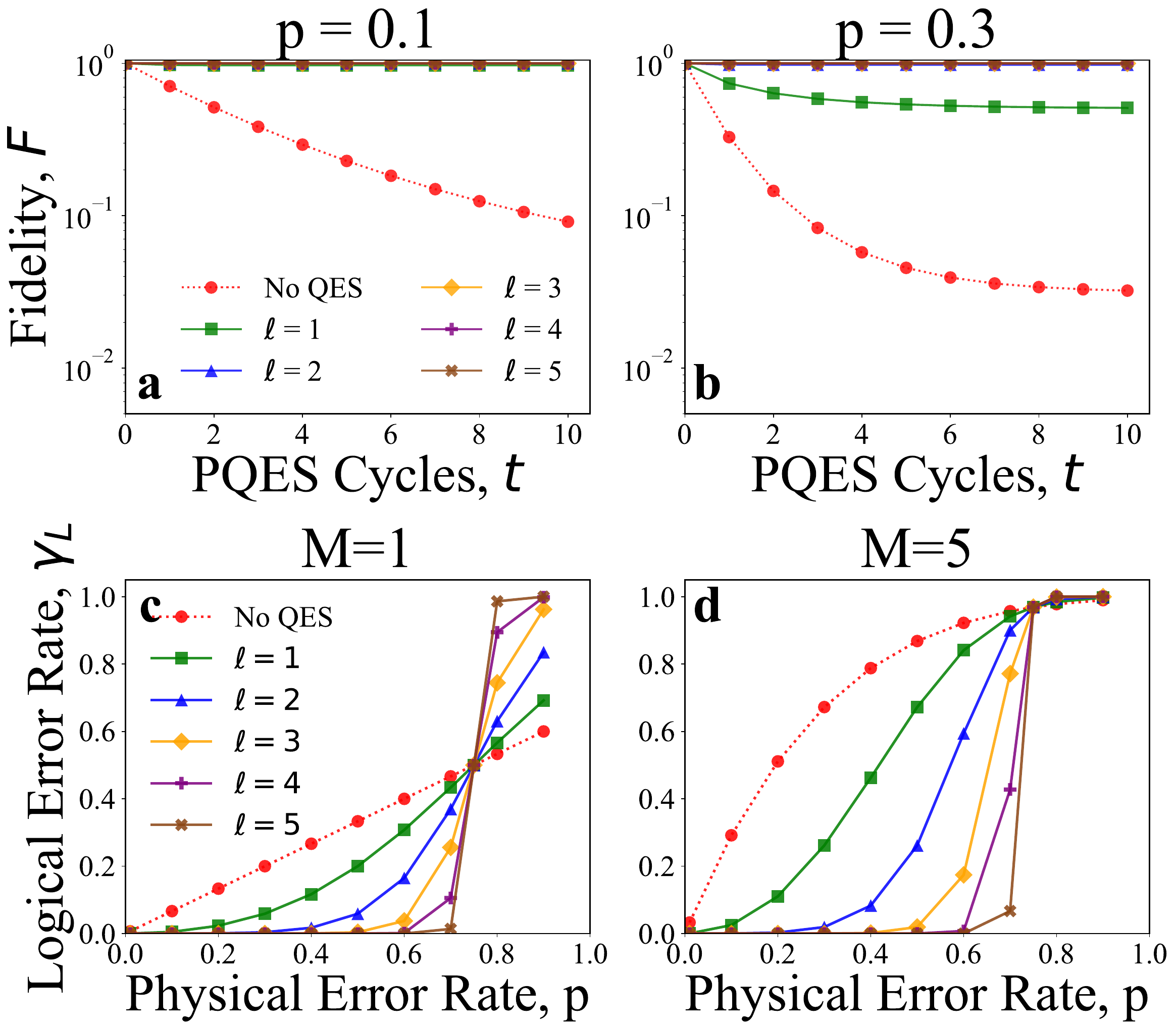}
    \caption{(Top) Fidelity evolution of locally depolarized states under PQES.  Plots show fidelity, $F$, versus cycles of PQES purification, $t$, for various numbers of purification rounds, $\ell$.  The initial state is the state $ | \psi_0 \rangle = \ket{+}^{\otimes M}$ with $M=5$,  subjected to the channel (\ref{eq:Dp-def}) with the error probability (a) $ p = 0.1 $; (b) $ p = 0.3 $.
    (Bottom) Effective logical error rate $\gamma_{L} $, versus physical error rate $p$, for local depolarizing noise (\ref{eq:Dp-def}) on the target state $|+\rangle^{\otimes M}$. Curves are shown for (c) $M=1$ and (d) $M=5$ for various purification rounds, $\ell$. The no-suppression baseline corresponds to $\ell=0$. }
    \label{fig:combined_fidelity_logical_error_depol}
\end{figure}


\section{Local anisotropic errors}
\label{sec:local_anisotropic_errors}

In this section, we analyze local anisotropic errors under PQES. Unlike depolarizing noise, an anisotropic channel can change the eigenbasis of the noisy state relative to the target state. Since PQES amplifies the dominant eigenspace of the noisy density matrix, rather than directly undoing the physical error channel, anisotropic errors can produce a dominant-eigenvector mismatch and hence a saturation of the achievable fidelity. We use local dephasing as a canonical example and then show how Clifford twirling can convert the anisotropic channel into an effective depolarizing channel that is more naturally handled by PQES.

\subsection{Error channel}

As a canonical example of an anisotropic error, we consider a dephasing channel which can be defined on the $m$th qubit as
\begin{equation}
\label{eq:Z_dephasing_channel_main}
    \mathcal{E}_m^{(z)} (\rho)
    = E_m^{(0)} \rho {E_m^{(0)}}^\dagger + E_m^{(1)} \rho {E_m^{(1)}}^\dagger
\end{equation}
where the Kraus operators are
\begin{align}
E_m^{(0)} & =\sqrt{1-p} I_m \nonumber \\
E_m^{(1)} & =\sqrt{p} Z_m .
\end{align}
On an $M$-qubit register, we take the noise to be 
%
\begin{equation}
    \mathcal{E}_z(\rho) =
    \left(\mathcal E^{(z)}_1\circ\mathcal E^{(z)}_2\circ\cdots\circ\mathcal E^{(z)}_M\right)(\rho).
\label{dephasingchannel}
\end{equation}
which we will refer to as a local dephasing model.

\subsection{Conversion of dephasing to depolarizing errors}
\label{subsubsec:twirl-channel}

As discussed in Sec.~\ref{sec:noisyqubit}, PQES is most directly effective when the noise preserves the direction of the target eigenvector and only reduces its spectral weight. Depolarizing noise has this property. Anisotropic noise need not: it can rotate or tilt the dominant eigenvector of the noisy state away from the original target, producing a coherent-mismatch floor for purification \cite{Huggins2021,Koczor2021_dominantEV}.

This effect is already visible for a single qubit. Using the Bloch-sphere parameterization (\ref{eq:noisysinglequbit}), the effect of applying the dephasing channel (\ref{dephasingchannel}) is to leave the $z$ component invariant while contracting the transverse components:
\begin{equation}
\label{eq:Z_error_bloch_contraction_mapping}
 \vec r \;\longmapsto\;
    ( (1-2p)  r_x, (1-2p)  r_y,  r_z).
\end{equation}
This results in a change in the direction of the Bloch vector (see Fig. \ref{diag:depol_bloch}).  Applying the purification step (\ref{eq:qubitradialupdate}) acts to repolarize the vector.  However, since the purification does not change the direction of the Bloch vector, the correction step does not revert the state to the original vector $ \vec{r}$ (see Appendix \ref{app:axis-dependence} for more details). The resulting asymptotic state is the dominant eigenvector of the dephased state, not necessarily the original target state. This is the origin of the fidelity saturation observed below.

\begin{figure}[t]
  \centering
  \includegraphics[width=\linewidth]{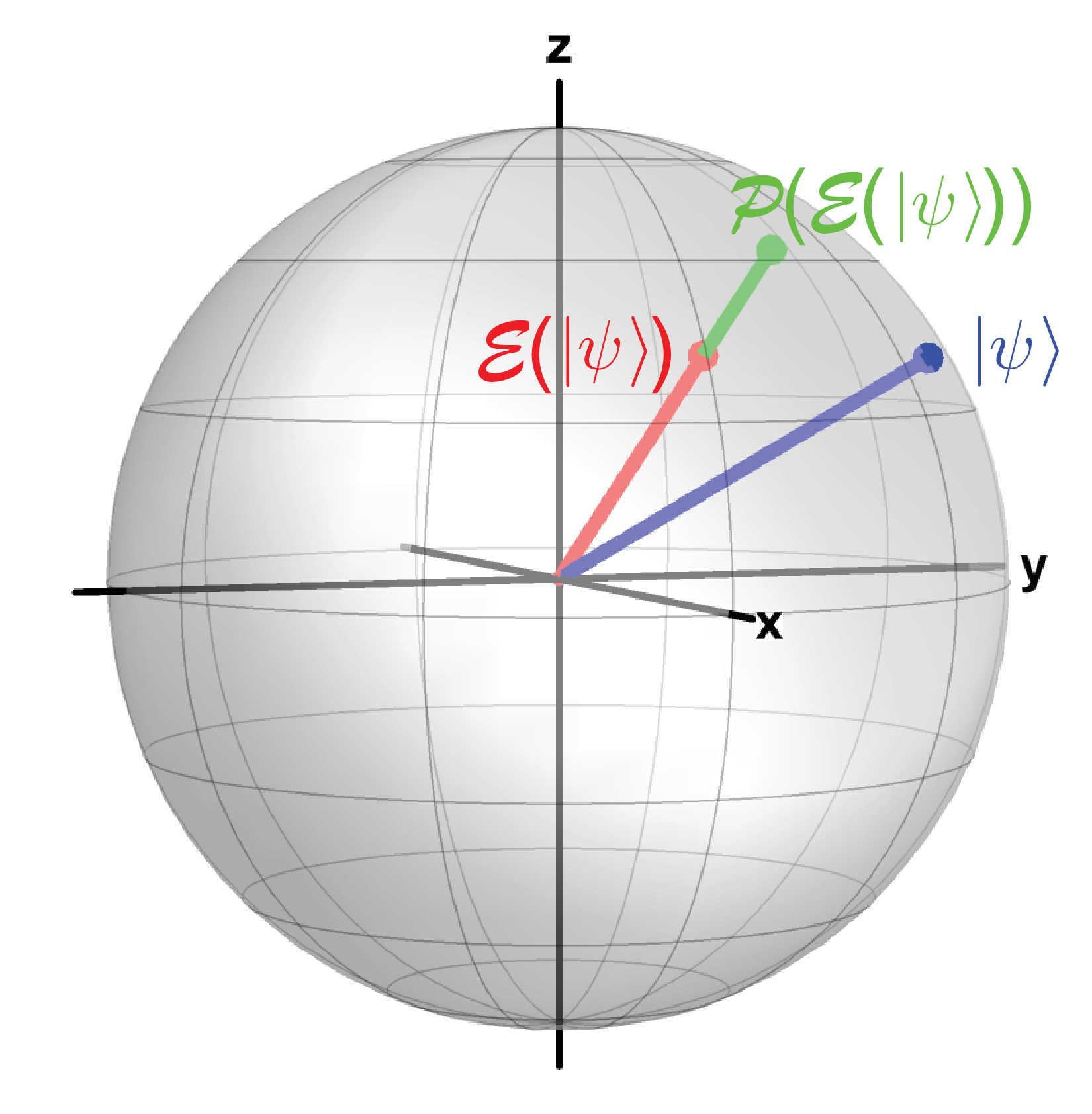}
  \caption{A single PQES cycle for a single-qubit state under dephasing.  The initial state $ | \psi \rangle $ has Bloch sphere parameters $ \theta = \pi/3, \phi = \pi/4 $.  The dephased state $ {\cal E} ( | \psi \rangle ) $ is applied according to  (\ref{dephasingchannel}) with $ p =0.3 $.  Finally one round of purification according to (\ref{eq:qubitradialupdate}) is applied to give ${\cal P} (  {\cal E} ( | \psi \rangle ) ) $.  The purified state lies in the same direction as $ {\cal E} ( | \psi \rangle ) $. }
  \label{diag:depol_bloch}
\end{figure}

We now show how Clifford frame-randomization rotations turn axis-biased dephasing into an effective depolarizing  channel, without changing the PQES map. We focus on local $Z$-dephasing (\ref{dephasingchannel}), but our results apply to any anisotropic local error channel.  The dephasing channel may be converted to a depolarizing channel by averaging over the $ X $, $ Y $, $ Z $-directions, by performing a basis transformation
\begin{align}
\label{eq:twirl-1qubit}
\mathcal{E}_m^{\mathrm{twirled}}(\rho) = \frac{1}{3} \sum_{U \in\{I,H,HS\}} 
U^\dagger {\cal E}_m^{(z)} (U \rho U^\dagger ) U .
\end{align}
where $ H $ is the Hadamard operator and $S=\mathrm{diag}(1,i)$  is the phase operator acting on the $ m$th qubit. Each $U$ permutes the Bloch axes, so averaging over the three rotations maps the local $ Z $ dephasing to a local depolarizing channel. On $M$ qubits, substituting (\ref{eq:twirl-1qubit}) into (\ref{localdepolchannelall}) gives the full local depolarizing channel. This would give $ 3^M $ different types of unitary operations to average over, corresponding to all combinations of $ \{I,H,HS\}^{\otimes M} $.  

In practice, the twirling process can be implemented by applying any one of the $ 3^M $ unitaries randomly before and after the dephasing errors take place. This converts the dephasing channel to a depolarizing channel such that the results of Sec. \ref{sec:local_depol} are recovered exactly. However, it has the drawback of averaging over a large number (e.g., $3^M$ for a deterministic implementation) of samples to obtain good statistics. This makes the full twirling process highly costly from both a real implementation and numerical point of view. 

More economically, a subset of the full set of $3^M$ gates can be applied for approximate twirling, while still improving results beyond the no-twirl baseline (see Sec. \ref{subsec:dephasing_error_threshold}). For this reason, we also examine approximate twirling protocols to convert local dephasing errors to a more isotropic form that has better performance with PQES. We consider the approximate twirling channel
\begin{align}
\label{eq:twirl-1qubit_approx}
\mathcal{E}_{\cal T} (\rho) = \frac{1}{T} \sum_{U \in {\cal T} } 
U^\dagger {\cal E}_z (U \rho U^\dagger ) U .
\end{align}
where $ T = |{\cal T}| $ and $ {\cal T} $ is a restricted set of Clifford unitaries to perform twirling over.

\begin{figure}[t]
  \centering
    \includegraphics[width=\linewidth]{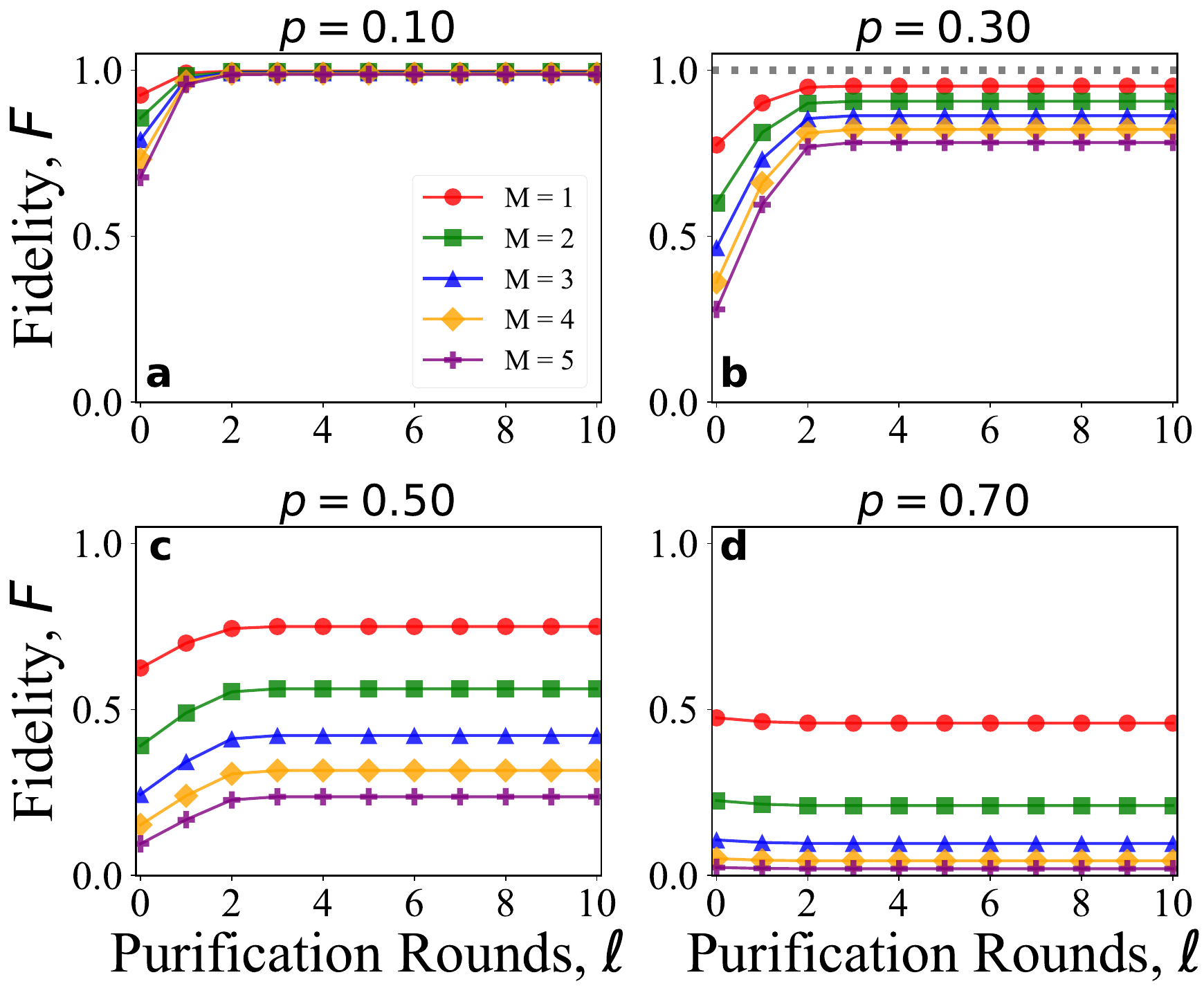}
    \caption{Fidelity improvement of locally dephased states under PQES, without twirling.  All plots show fidelity, $F$, versus rounds of purification, $\ell$, for various qubit numbers $ M $.  The initial state is $|\psi_0 \rangle = (\cos \tfrac{\theta}{2} |0 \rangle + e^{i \phi} \sin \tfrac{\theta}{2} |1 \rangle  )^{\otimes M} $, with parameters $\theta=\pi/3,  \phi=\pi/4$, and is subjected to the channel (\ref{eq:Z_dephasing_channel_main}) with the error probability (a) $ p = 0.1 $; (b) $ p = 0.3 $; (c) $ p = 0.5 $; (d) $ p = 0.7 $.
    }
    \label{fig:fidelity_grid_dephasing_untwirled}
\end{figure}

\subsection{Purification}

Figure \ref{fig:fidelity_grid_dephasing_untwirled} shows the fidelity evolution under multiple rounds of purification for the state $|\psi_0 \rangle = (\cos \tfrac{\theta}{2} |0 \rangle + e^{i \phi} \sin \tfrac{\theta}{2} |1 \rangle  )^{\otimes M} $, for various $ M $ and physical dephasing error probabilities $ p $.  No twirling is performed.  We see that, as before, below a critical error probability, the purification acts to improve the fidelity.  The primary difference from the depolarizing case is that the fidelities saturate to a value less than $ F = 1 $. This is the dominant-eigenvector mismatch illustrated in Fig.~\ref{diag:depol_bloch}: the dephasing channel changes the direction of the Bloch vector, and PQES repolarizes the state along this new direction. However, for small error probabilities, the fidelities nevertheless converge to values close to unity, showing that the purification can still be effective, although not perfect. For error probabilities above $ p>1/2 $, we do not expect good performance, since at $ p = 1/2 $, the $ x $ and $ y $ components of the Bloch vector (\ref{eq:Z_error_bloch_contraction_mapping}) are completely erased.  We see in Fig. \ref{fig:fidelity_grid_dephasing_untwirled}(c) the fidelities do, in fact, increase because the purification converges to $ | 0 \rangle^{\otimes M} $, which is closer to the original state than the initially dephased state.  

To overcome the saturation effect, we apply the twirling operation of (\ref{eq:twirl-1qubit}) to each qubit by averaging over $ 3^M $ unitary operations. According to the discussion in Sec. \ref{subsubsec:twirl-channel}, this converts the local dephasing channel to a local  depolarizing channel with the same parameter $p$.  This prediction is borne out numerically.  Performing the full twirling with the dephasing channel reproduces precisely the same graphs as shown in the genuine depolarizing channel Fig.~\ref{fig:fidelity_grid_depolarizing} (we do not show the twirled plots as they are visually identical). Interpreting Fig. \ref{fig:fidelity_grid_depolarizing} as the results for twirled dephasing and comparing to Fig. \ref{fig:fidelity_grid_dephasing_untwirled}, we see that the twirling generally acts to improve the fidelity in a wider range $p<3/4$.  For small errors (e.g. $ p=0.1 $) there is only a marginal improvement but for larger errors the increase in fidelity is more marked.  

Figure \ref{fig:SWAP_thresholds}(c)-(f) shows the shift in the critical error probability with and without twirling under the dephasing channel.  The plots show the fidelity $F$ versus $p$ for various $\ell$ and system sizes $M$. With the initial state set at $ | + \rangle^{\otimes M} $, and applying dephasing without twirling, additional rounds of purification act to improve the fidelity only if $ p< 1/2 $.  When twirling is added, the critical point shifts to $ p =3/4 $.  The critical point is unchanged for different $ M $. However, for larger $M$, the fidelity for the same number of purification rounds is reduced.  We may attribute this to the fact that for the same error rate $ p$, a larger system tends to have a higher probability that there is an error.  For example, the probability that there is no error in the whole system under the channel (\ref{dephasingchannel}) is $ (1-p)^M $, which reduces exponentially with $ M $.

\begin{figure}[t]
  \centering
    \includegraphics[width=\linewidth]{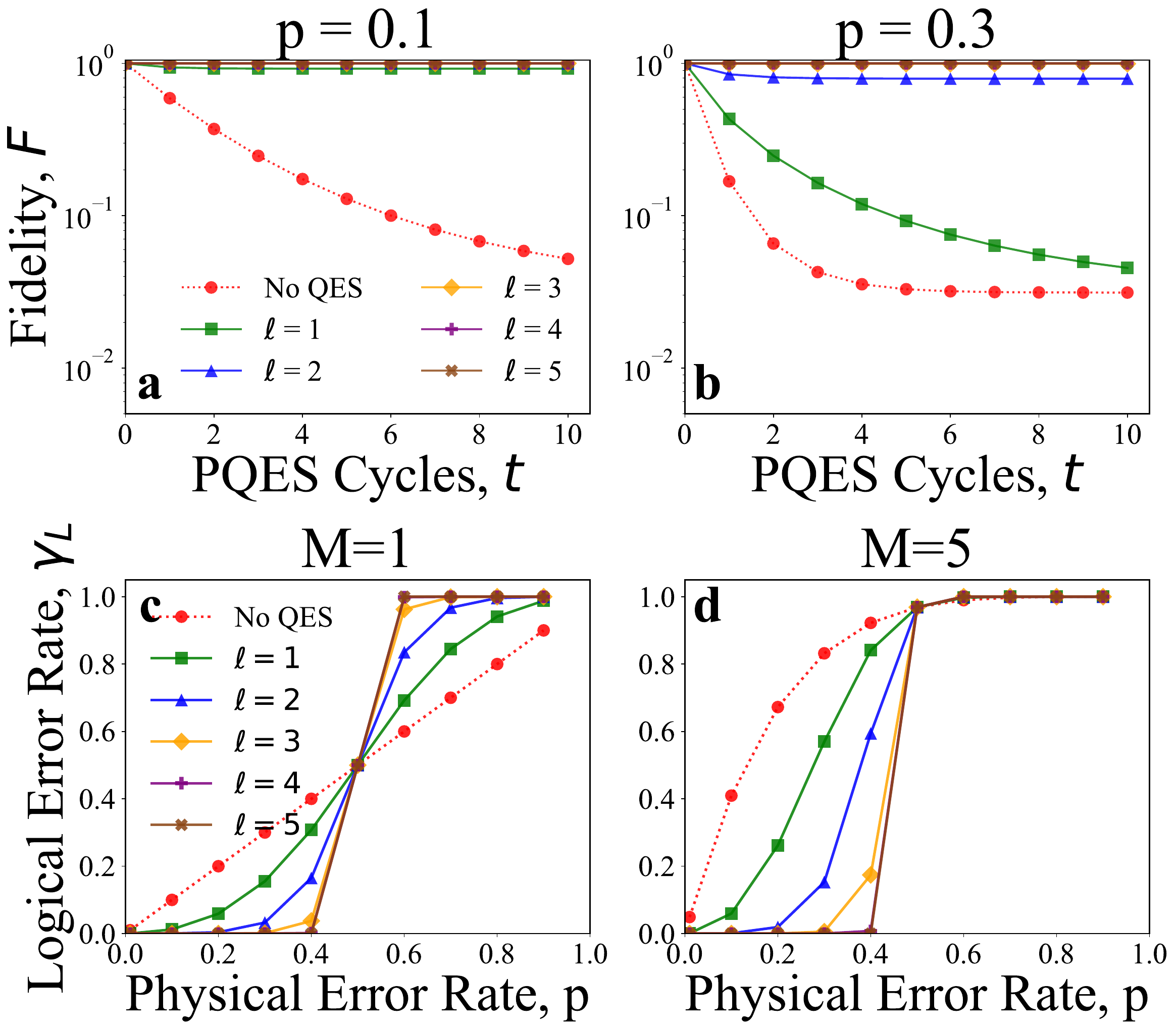}
    \caption{(Top) Fidelity improvement of locally dephased states with no twirling under PQES.  Plots show fidelity, $F$, versus iterations of PQES purification, $t$, for various numbers of purification rounds, $\ell$.  The initial state is the state $ | \psi_0 \rangle = \ket{+}^{\otimes M}$ with $M=5$,  subjected to the channel (\ref{dephasingchannel}) with the error probability (a) $ p = 0.1 $; (b) $ p = 0.3 $.
    (Bottom) Effective logical error rate, $\gamma_L$, versus physical error rate, $p$, for pure dephasing noise without twirling on the target state $|+\rangle^{\otimes M}$ with noise channel (\ref{dephasingchannel}). Curves are shown for $M=1$ (c) and $M=5$ (d) for various purification rounds, $\ell$. The no-suppression baseline corresponds to $\ell=0$. The threshold is seen at $p_{\mathrm{th}}=0.5$.}
    \label{fig:combined_fidelity_logical_error_pure_deph}
\end{figure}

\begin{figure}[t]
  \centering
    \includegraphics[width=\linewidth]{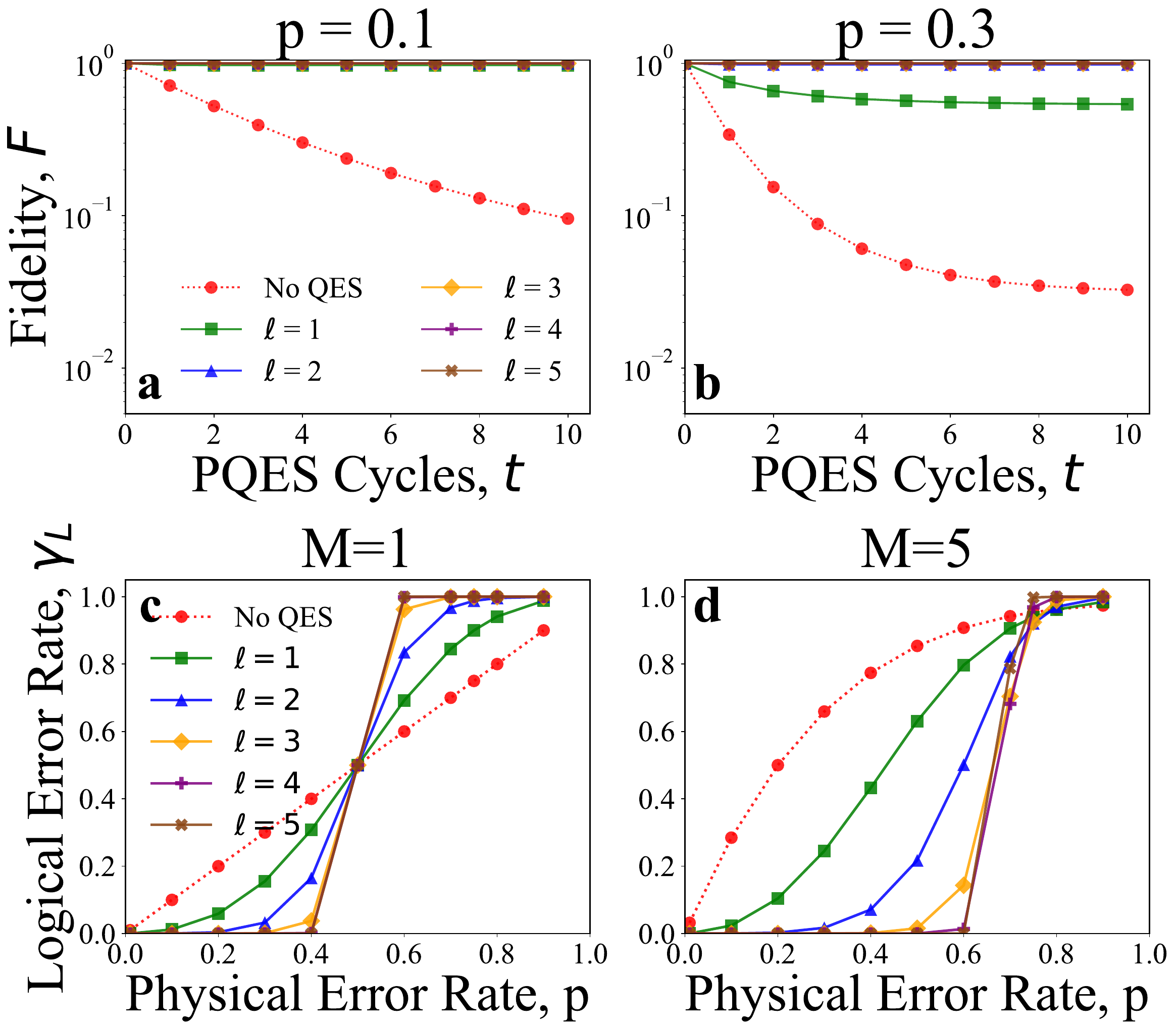}
    \caption{(Top) Fidelity improvement of locally dephased states with approximate twirling under PQES.  Plots show fidelity, $F$, versus iterations of PQES purification, $t$, for various numbers of purification rounds, $\ell$.  The initial state is the state $ | \psi_0 \rangle = \ket{+}^{\otimes M}$ with $M=5$,  subjected to the channel (\ref{eq:twirl-1qubit_approx}) with the error probability (a) $ p = 0.1 $; (b) $ p = 0.3 $.
    (Bottom) Effective logical error rate, $\gamma_L$, versus physical error rate, $p$, for approximately twirled-dephasing noise on the target state $|+\rangle^{\otimes M}$ with noise channel (\ref{eq:twirl-1qubit_approx}). Curves are shown for $M=1$ (c) and $M=5$ (d) for various purification rounds, $\ell$. The no-suppression baseline corresponds to $\ell=0$. Approximate twirling corresponds to using $20\%$ of the total $3^M$ total gate combinations, rounded up to a whole number.}
    \label{fig:combined_fidelity_logical_error_subsetTwirl_deph}
\end{figure}

\subsection{PQES error thresholds}
\label{subsec:dephasing_error_threshold}

We now consider the error threshold analysis by following the procedure in Sec. \ref{subsec:threshold-performance} using the local dephasing channel (\ref{dephasingchannel}) starting from the initial state $ |\psi_0 \rangle = |+\rangle^{\otimes M} $. 

Figure \ref{fig:combined_fidelity_logical_error_pure_deph} shows results for local dephasing without twirling for $M=1$ and $M=5$. Figure \ref{fig:combined_fidelity_logical_error_pure_deph}(a)(b) shows the fidelity evolution as a function of error and purification cycles for $ M = 5 $.  Similarly to the global and local depolarizing channels (see Figs. \ref{fig:onequbitrec} and \ref{fig:combined_fidelity_logical_error_depol}), the fidelity follows an exponential decay that saturates to a finite value for the physical error values shown. As before, increasing the purification depth, or equivalently the number of copies $N=2^\ell$, reduces the effective logical error rate (\ref{gammadef}).

Figure \ref{fig:combined_fidelity_logical_error_pure_deph}(c)(d) shows the effective logical error rate $ \gamma_{L} $ as a function of the physical error rate $ p $.  For both $M=1$ and $M=5$, we observe a  crossover of the curves, giving the PQES error threshold at 
\begin{align}
    p_{\mathrm{th}} = 1/2 .
\end{align}
For $p<p_{\mathrm{th}}$, increasing the number of purification rounds $\ell$ substantially suppresses the effective logical error rate $\gamma_L$, reflecting that each PQES cycle yields a net purification gain that slows logical error accumulation over time. As $p\to p_{\mathrm{th}}$, the marginal benefit of increasing $\ell$ diminishes, and progressively larger $\ell$ is required to maintain a small effective decay rate. For $p>p_{\mathrm{th}}$, additional purification rounds amplify the wrong dominant eigenvector, so the target fidelity decreases rather than improves.

With full twirling, (\ref{eq:twirl-1qubit}) maps local dephasing to the local depolarizing channel with the same contraction parameter. The PQES dynamics then coincide pointwise with the local depolarizing results of Fig.~\ref{fig:combined_fidelity_logical_error_depol}: the fidelity relaxation, steady-state values, and effective decay rates are the same, and the threshold is restored to $p_{\mathrm{th}}=3/4$.

We also consider approximate twirling, where only a fraction of the full $3^M$ Clifford gates is applied. Figure \ref{fig:combined_fidelity_logical_error_subsetTwirl_deph} shows results for local dephasing with only $20\%$ twirling applied for $M=1$ and $M=5$ (rounded up to a whole number of gates); the applied twirling gates are randomly selected from the total set of $3^M$ options. Figure \ref{fig:combined_fidelity_logical_error_subsetTwirl_deph}(a)(b) shows the fidelity evolution as a function of error and purification cycles for $ M = 5 $.  Similarly to the global and local depolarizing channels (see Figs. \ref{fig:onequbitrec} and \ref{fig:combined_fidelity_logical_error_depol}), the fidelity follows an exponential decay that saturates to a finite value for the physical error values shown. As before, increasing the purification depth, or equivalently the number of copies $N=2^\ell$, reduces the effective logical error rate (\ref{gammadef}). 

Figure \ref{fig:combined_fidelity_logical_error_subsetTwirl_deph}(c)(d) shows the effective logical error rate $ \gamma_L $ as a function of the physical error rate $ p $. For $M=1$, approximate twirling here does not change anything from the untwirled case, as it corresponds to only applying a single Clifford gate; on average, this simply rotates the dephasing axis but does not make the noise channel any less anisotropic. Hence, the crossover is $p_{\mathrm{th}}=0.5$, identical to pure dephasing without twirling, as shown in Fig~\ref{fig:combined_fidelity_logical_error_pure_deph}(c). For $ M= 5$, the crossover is consistent with the PQES error threshold  $p_{\mathrm{th}} \sim 0.75$. For $p<p_{\mathrm{th}}$, increasing the number of purification rounds $\ell$ substantially suppresses the effective logical error rate $\gamma_L$, reflecting that each PQES cycle yields a net purification gain that slows logical error accumulation over time. As $p\to p_{\mathrm{th}}$, the marginal benefit of increasing $\ell$ diminishes, and progressively larger $\ell$ is required to maintain a small effective decay rate. For $p> p_{\mathrm{th}}$, adding additional purification rounds produces a worse result, characteristic of the error rate being over the threshold.  

We see that approximate twirling is effective at increasing the PQES error threshold beyond the no-twirl crossover at $p_{\mathrm{th}}=0.5$, with full twirling achieving $p_{\mathrm{th}}=0.75$, consistent with the isotropic (depolarization) case. One can accordingly tune the number of twirling gates applied, in order to negotiate the threshold with the resource demands of applying more gates.

In both noise families, the qualitative picture is therefore the same: below the relevant threshold, PQES can, in principle, drive the single-layer fidelity arbitrarily close to unity by increasing $\ell$ (in the absence of additional noise within the PQES layer). However, above the threshold, the protocol cannot asymptotically recover the target state, even with additional purification rounds.

\section{Summary and Conclusions}
\label{sec:conclusion}

We have introduced a scheme for purification-based quantum error suppression using SWAP tests on multiple noisy copies, which we refer to as purification quantum error suppression (PQES). Operationally, PQES uses multiple noisy copies of the same state rather than an encoded code block. Given $ N $ identically prepared noisy copies of a state, it recursively compares and merges pairs. In the non-postselected scheme developed here, all SWAP-test outcomes are retained and combined with signs, yielding $\rho^N$ as physically present in the final output state. This then provides access to expectation values with respect to the purified state $\rho^N/\Tr(\rho^N)$. PQES suppresses the error of an {\it unknown} principal eigenstate of the noisy density matrix by consuming additional noisy copies and processing the resulting SWAP-outcome record. Notably, these properties are QEC-like: the quantum registers undergo physical measurement-induced transformations, no branch is discarded, the purification structure may be distributed through an algorithmic circuit, and the number of copies $N$ plays the role of the scalable redundancy parameter. The fact that this purification works on unknown states is crucial for circuit applications, where the intermediate quantum state is generally not known.



The mechanism of PQES is that the power map amplifies the dominant eigenvector of the noisy state. As a result, PQES is most naturally suited for depolarizing errors, where the PQES error threshold was evaluated to be $ 75 \% $ for local errors. Though a direct comparison with typical QEC thresholds is not possible due to the distinct mechanisms and operations involved, we nevertheless emphasize that the PQES error threshold here is notably higher than what is possible with conventional QEC codes. For example, the code-capacity thresholds for stabilizer codes have been reported to be $18.9\% $ \cite{dennis2002topological,wang2011threshold,bombin2012strong,wootton2012high}, with an $25 \%$ upper bound due to the no-cloning theorem \cite{smith2006upper}.

Compared to depolarizing errors, local anisotropic errors behave differently because they can change the dominant eigenvector of the noisy state. For local dephasing, the purification map repolarizes along the dephased Bloch direction (Fig. \ref{diag:depol_bloch}) and therefore generally saturates below unit fidelity when the dominant eigenvector is misaligned with the original target. Regardless, the difference in performance for small error probabilities is not drastically different between depolarizing and dephasing errors. Furthermore, full local twirling converts the dephasing channel into an effective local depolarizing channel and restores the corresponding depolarizing threshold. These thresholds identify the point at which the target state ceases to be the dominant eigenvector of the noisy density matrix. Below this point, increasing the purification depth suppresses the effective logical decay in the repeated-cycle model. Above it, the same power map converges to a different dominant eigenvector or eigenspace; additional purification rounds cannot recover the target.


The resource requirements of PQES depend strongly on the implementation. In the fully parallel binary-tree realization of Fig.~\ref{diag:main_use_case}(a), the coherent data footprint scales as $n \sim O(M 2^\ell)$ for an $M$-qubit register and $N=2^\ell$ copies, while the SWAP-test layers have depth $O(\ell)$ up to the architecture-dependent cost of a register-level controlled-SWAP. At first glance, this would appear to have a heavy qubit count, but errors are also suppressed exponentially with $ N $, due to (\ref{eq:purifiedstate}). Then the logical error rate scales as $ \gamma_L \sim e^{-N} \sim e^{- n/M } $. This is comparable with the scaling of standard QEC, where there is exponential scaling of logical errors with the code distance. For example, in the surface code, the number of physical qubits scales as the square of the code distance so that $ \gamma_L \sim e^{- \sqrt{n} } $. The scaling in terms of the number of qubits can be even improved further by using the recycled implementation of Sec.~\ref{sec:protocol}; here, the coherent data footprint is reduced to $O(M\ell)$ by generating the tree sequentially and keeping only one node per depth live at a time. The tradeoff is an increased sequential SWAP-test count of $N-1$. This implementation occupies a different architectural regime from the smallest known reset-based virtual-distillation constructions~\cite{czarnik_qubit_eff} --- it provides a pairwise-SWAP architecture with an explicit outcome record that is natural for a continuous input of copies, register recycling, explicit outcome-record processing, and interleaved error-suppression schedules.

The exponential suppression obtained from the power map must also be balanced against the sampling overhead of the signed estimator. Since observables are estimated as ratios involving $\Tr(\rho^N)$, highly mixed states and large purification powers require more samples. Thus, PQES has the characteristic tradeoff of QEM: increasing $N$ improves spectral selectivity, but the normalization of $\Tr(\rho^N)$ controls the number of samples needed to resolve the mitigated observable. Just as in QEM, PQES suppresses errors in estimated observables by consuming additional copies and samples rather than protecting a single encoded logical state throughout a computation with polynomial fault-tolerant overhead.

As the PQES redundancy consists of corresponding noisy copies rather than an entangled code block, an ideal circuit layer is applied as the same unitary on each copy, which makes common algorithmic gates structurally simple. At a stage containing $k$ corresponding copies, an ideal circuit block is implemented as $U^{\otimes k}$. Each factor acts only within one copy register, so an error arising during one application of $U$ does not propagate directly into the other copies before the SWAP-test layer. This differs from encoded stabilizer-code computation, where logical gates act on an entangled code space and the ability to perform transversal gates is constrained by the Eastin--Knill theorem~\cite{eastin2009restrictions}. While we leave questions regarding fault-tolerant PQES as future work, the simple transversal structure of logical gates is promising in terms of suppressing error proliferation. 


In this paper, our focus has been on the quantum-computing scenario, but the same multi-copy purification primitive may be useful in other settings where many nominally identical noisy quantum states are produced. Examples include high-fidelity state preparation, Bell-pair or magic-state factories, quantum sensor networks~\cite{Eldredge2018}, distributed quantum computing~\cite{Kimble_2008,Illiano_2022}, quantum metrology with spin ensembles~\cite{Wineland1992,Ma2025,grafe2025ultrahigh}, and quantum communication or repeater architectures~\cite{Briegel1998,Bennett1996Twirling}. In such regimes, PQES can act as a continuous purification layer, distilling higher-quality observable estimates from noisy input copies while keeping the coherent memory footprint small.


\medskip
\begin{acknowledgments}
This work is supported by the SMEC Scientific Research Innovation Project (2023ZKZD55); the National Natural Science Foundation of China (92576102); the Science and Technology Commission of Shanghai Municipality (22ZR1444600); the NYU Shanghai Boost Fund; the China Foreign Experts Program (G2021013002L); the NYU-ECNU Institute of Physics at NYU Shanghai; the NYU Shanghai Major-Grants Seed Fund; and Tamkeen under the NYU Abu Dhabi Research Institute grant CG008.
\end{acknowledgments}

\paragraph*{Data and Code Availability.}
All data supporting the findings of this study are available from the corresponding author upon reasonable request.

\bibliography{references}

\clearpage

\appendix


\section{PQES Overhead}
\label{app:overhead}

Eq.~\eqref{purifiedoexp} is evaluated using experimental shots $k\in\{1,\dots,N_{\mathrm{samp}}\}$. From each shot, we obtain (i) a SWAP-outcome string $\vec{\sigma}^{(k)}_\ell$ and hence a parity weight $\Omega_k:=\Omega_{\vec{\sigma}^{(k)}_\ell}\in\{\pm1\}$,
and (ii) a single-shot measurement outcome $o_k$ from measuring the observable $O$ on the output register (i.e. $o_k\in\mathrm{spec}(O)$). Define the sample means

\begin{equation}
\widehat{A}:=\frac{1}{N_{\mathrm{samp}}}\sum_{k=1}^{N_{\mathrm{samp}}}\Omega_k\,o_k,
\quad
\widehat{B}:=\frac{1}{N_{\mathrm{samp}}}\sum_{k=1}^{N_{\mathrm{samp}}}\Omega_k,
\end{equation}
so that $\mathbb{E}[\widehat{A}]=\Tr(O\rho^N)$ and $\mathbb{E}[\widehat{B}]=\Tr(\rho^N)$, and we estimate $\langle O\rangle_\ell$ by $\widehat{\langle O\rangle}_\ell:=\widehat{A}/\widehat{B}$. By a standard delta-method (large-$N_{\mathrm{samp}}$) expansion for ratio estimators, the asymptotic variance is
\begin{equation}
\mathrm{Var}\!\bigl(\widehat{\langle O\rangle}_\ell\bigr)
\;\approx\;
\frac{1}{N_{\mathrm{samp}}\,\Tr(\rho^N)^2}\;
\mathrm{Var}\!\Big(\Omega_{\vec{\sigma}_{\ell}}  \,[\,\langle O \rangle_{\vec{\sigma}_{\ell}}-\langle O\rangle_\ell\,]\Big),
\label{eq:vd_ratio_var}
\end{equation}
and hence the standard error is
\begin{equation}
\epsilon
\;\approx\;
\frac{\sqrt{\mathrm{Var}\!\big(\Omega_{\vec{\sigma}_{\ell}} \,\bigl[\,\langle O \rangle_{\vec{\sigma}_{\ell}}-\langle O\rangle_\ell\,\bigr]\big)}}{\sqrt{N_{\mathrm{samp}}}\;|\Tr(\rho^N)|}.
\label{eq:vd_ratio_se}
\end{equation}
Since $\Omega^2=1$ and $\mathbb{E}[\Omega(o-\langle O\rangle_\ell)]=0$, one also has the simplification
$\mathrm{Var}\!\big(\Omega[\,o-\langle O\rangle_\ell\,]\big)=\mathbb{E}[(o-\langle O\rangle_\ell)^2]$.
In particular, if $O$ is normalized so that $\|O\|_\infty\le 1$ (e.g. Pauli strings), then
$\mathrm{Var}(\Omega[\,o-\langle O\rangle_\ell\,])\le 4$ giving 
\begin{equation}
\epsilon \;\lesssim\; \frac{2}{\sqrt{N_{\mathrm{samp}}}\;|\Tr(\rho^N)|}. 
\end{equation}

\section{Pauli basis expansion}
\label{app:paulibasis}

We connect the spectral action of the PQES map from Sec. \ref{sec:purification} to the Pauli-basis expansion. First expand mixed states using the $M$-qubit Pauli expansion
\begin{equation}
    \rho = \frac{1}{2^M}\sum_{P} r_P P,
\end{equation}
where the sum runs over all $4^M$ Pauli strings $P\in\{I,X,Y,Z\}^{\otimes M}$ and
$r_P := \Tr(\rho P)\in\mathbb{R}$. In these coordinates, the purity is
\begin{equation}
    \Tr(\rho^2) = \frac{1}{2^M}\sum_{P} r_P^2,
\end{equation}

Depolarizing noise contracts all non-identity Pauli components uniformly, $r_P \mapsto (1-p)\,r_P$ for $P\neq I$, producing an \emph{isotropic} shrinkage of the generalized Bloch vector and a corresponding reduction in purity.
By contrast, the PQES map ${\cal P}$ does not, in general, act as an isotropic (radial) expansion in Pauli space. Indeed, although $\Tr(\rho^2)=2^{-M}\sum_{P} r_P^2$ implies that ${\cal P}$ increases the Euclidean radius $\sum_{P\neq I} r_P^2$ whenever it increases purity, the map ${\cal P}$ is nonlinear and spectral: it preserves the eigenbasis and updates only the eigenvalues according to \eqref{eq:eigenvalue_update}. An isotropic rescaling $r_P \mapsto \alpha\,r_P$ (with a single $\alpha$ for all $P\neq I$) would correspond to an affine, rotationally invariant action on operator space, whereas \eqref{eq:eigenvalue_update} shows that ${\cal P}$ depends on the full set $\{\lambda_i\}$ through the normalization $\sum_j \lambda_j^2$, and therefore generically changes not only the length but also the direction of the Pauli-coefficient vector $(r_P)_{P\neq I}$.

Nevertheless, in the isotropic regime, the dynamics are confined to highly symmetric families of spectra, so the action of ${\cal P}$ is well captured by a one-parameter ``radial'' picture:
depolarizing (or fully twirled dephasing) drives the spectrum toward uniformity by shrinking $\sum_j \lambda_j^2$, while purification counteracts this by polarizing the spectrum (increasing the contrast among the $\lambda_i$ via \eqref{eq:ratio_increase}), which necessarily increases $\Tr(\rho^2)=\sum_i \lambda_i^2$ and hence increases $\sum_{P\neq I} r_P^2$.
This explains why our protocol is naturally well matched to depolarizing noise and why fully twirled dephasing exhibits identical performance: the noise acts isotropically in Pauli space, and purification reliably increases the Bloch radius, even though it is not, in general, an isotropic expansion map on the individual $r_P$ components.

\section{Purification and purity}
\label{app:purification}

Here we show that the normalized power map used in PQES increases the purity of a state and characterize its fixed points. For one purification step,
\begin{align}
{\cal P}(\rho)=\frac{\rho^2}{\Tr(\rho^2)} .
\end{align}
More generally, the $\ell$-round purification map is ${\cal P}_\ell(\rho)=\rho^N/\Tr(\rho^N)$ with $N=2^\ell$.

Let
\begin{align}
\rho = U\,\mathrm{diag}(\boldsymbol{\lambda})\,U^\dagger,
\qquad
\boldsymbol{\lambda}=(\lambda_i),
\end{align}
where $\lambda_i\ge0$ and $\sum_i\lambda_i=1$. Then one purification step gives
\begin{equation}
\label{eq:eigenvalue_update}
   \rho' ={\cal P}(\rho)
    = U\,\mathrm{diag}(\boldsymbol{\lambda}')\,U^\dagger,
    \qquad
    \lambda_i'
    = \frac{\lambda_i^2}{\sum_j \lambda_j^2}.
\end{equation}
Thus the map is spectral: the eigenbasis is unchanged and only the spectrum is updated. Because $g(x)=x^2$ is strictly increasing on $[0,1]$, the ordering of eigenvalues is preserved. Moreover, for $\lambda_i>\lambda_j>0$,
\begin{align}
\frac{\lambda_i'}{\lambda_j'}
=
\left(\frac{\lambda_i}{\lambda_j}\right)^2
>
\frac{\lambda_i}{\lambda_j},
\end{align}
so the relative contrast between unequal nonzero eigenvalues increases.

The purity after one purification step is
\begin{equation}
\label{eq:purity_monotone}
    \Tr({\rho'}^2)
    = \sum_i (\lambda_i')^2
    = \frac{\sum_i \lambda_i^4}{\bigl(\sum_j \lambda_j^2\bigr)^2}.
\end{equation}
Define the power sums
\begin{align}
s_n:=\sum_i\lambda_i^n .
\end{align}
Then
\begin{align}
\Tr({\rho'}^2)=\frac{s_4}{s_2^2},
\qquad
\Tr(\rho^2)=s_2,
\end{align}
and hence
\begin{equation}
\Tr({\rho'}^2)-\Tr(\rho^2)
=
\frac{s_4-s_2^3}{s_2^2}.
\label{purityinc}
\end{equation}
It remains to show that $s_4\ge s_2^3$. This follows directly from Jensen's inequality. Treating the eigenvalues $\lambda_i$ as a probability distribution and applying Jensen's inequality to the convex function $x^3$ gives
\begin{align}
s_4
=
\sum_i \lambda_i \lambda_i^3
\ge
\left(\sum_i \lambda_i \lambda_i\right)^3
=
s_2^3 .
\end{align}
Therefore,
\begin{align}
\Tr({\rho'}^2)\ge \Tr(\rho^2).
\end{align}
Equality in Jensen's inequality holds if and only if all nonzero eigenvalues of $\rho$ are equal. That is, equality holds precisely when $\rho$ is maximally mixed on its support. These states are exactly the fixed points of the normalized power map. Pure states are included as the rank-one case: if $\rho$ has purity one, then $\rho^2=\rho$ and the map leaves it unchanged. The full-rank maximally mixed state is the opposite extreme.

For a full-rank state, the only mixed fixed point is the maximally mixed state. For a rank-deficient state, normalized projectors onto lower-dimensional subspaces are also fixed points. Thus, except for flat spectra on their support, the purification map strictly increases the purity and sharpens the spectrum toward the dominant eigenspace.

\section{Fidelity convergence under SWAP purification}
\label{app:fidelity_convergence}

In this section we show monotonic improvement and convergence of the fidelity $F \to 1$ under PQES for the global depolarizing channel of Sec. \ref{sec:global_depol}.  

For the Werner state (\ref{wernerrho}), $\rho$ commutes with $\rho_0=\ket{\psi}\!\bra{\psi}$. Then the spectrum of $\rho$ is 
\begin{equation}
    \lambda_1=F,\qquad \lambda_2=\cdots=\lambda_D=\frac{1-F}{D-1}. 
\end{equation}
It follows that $\Tr\rho^2=F^2+\frac{(1-F)^2}{D-1}$ and $\bra{\psi}\rho^2\ket{\psi}=F^2$. Under the PQES update (\ref{eq:purifiedstaten2}) the fidelity update is therefore
\begin{equation}
\label{eq:gF-def}
    F'=g(F):=
    \frac{F^2}{F^2+\dfrac{(1-F)^2}{D-1}},\qquad D\ge 2.
\end{equation}
This recursive relation has the following properties. 

\begin{theorem}[Monotone behavior and convergence in the isotropic family]
\label{thm:isotropic-conv}
For any $D\ge 2$, the recursion $F'=g(F)$ in \eqref{eq:gF-def} has fixed points at $F_\star=1$ and $F_\star=1/D$ (and also $F_\star=0$ for the extended recursion on $F\in[0,1]$).
Moreover:
(i) if $F>1/D$ then $F'$ is strictly increasing and converges to $1$;
(ii) $F=1/D$ is a fixed point (the maximally mixed state in the Werner family).
The fixed point at $F=1$ is locally attractive with \emph{quadratic} rate.
\end{theorem}

\begin{proof}

\emph{(i) $g(F)>F$ for $F\in(1/D,1)$:} For $F\in(0,1)$, compute
\begin{align}
    g(F)-F
    &= \frac{F^2}{F^2+\frac{(1-F)^2}{D-1}} - F \nonumber\\
    &= \frac{F(1-F)\,(DF-1)}{(D-1)F^2+(1-F)^2}.
\end{align}
The denominator is strictly positive and $F(1-F)>0$ on $(0,1)$, hence $\mathrm{sign}(g(F)-F)=\mathrm{sign}(DF-1)$.
Therefore $g(F)>F$ iff $F>1/D$, with equality at $F\in\{1/D,1\}$.

\emph{(ii) Boundedness:} $g(F)\in[0,1]$ for $F\in[0,1]$, and $g(F)<1$ unless $F=1$, since the denominator in \eqref{eq:gF-def} exceeds the numerator unless $1-F=0$.

\emph{(iii) Convergence and fixed point:} If $F_0>1/D$, then $F_{n+1}=g(F_n)>F_n$ and the sequence is monotone increasing and bounded above by $1$, hence convergent to some limit $L\le 1$.
By continuity, $L=g(L)$, so $L$ must be a fixed point. Since $L\ge F_0>1/D$ and $g(F)>F$ on $(1/D,1)$, the only consistent limit is $L=1$.
If $F_0=1/D$, then $F_n\equiv 1/D$ for all $n$.

\emph{(iv) Local rate near $F=1$:} Write $F=1-\varepsilon$ with $\varepsilon\ll 1$. Then
\begin{align}
    F'
    &=\frac{(1-\varepsilon)^2}{(1-\varepsilon)^2+\frac{\varepsilon^2}{D-1}}
    =1-\frac{\varepsilon^2}{D-1}+O(\varepsilon^3),
\end{align}
so $1-F'$ is quadratic in $\varepsilon$. This shows superlinear (quadratic) convergence to $F=1$.
\end{proof}

\section{Small error expansion for local depolarizing channel}
\label{app:smallerror}

Consider the Pauli expansion of an arbitrary pure state $\rho_0=\ket{\psi_0}\!\bra{\psi_0}$ as given in~(\ref{pauliexp}). Define the weight $w(P)$ of a Pauli string $P$ as the number of non-identity single-qubit factors in $P$:
\begin{equation}
    w(P) := \bigl|\{ m \in \{1,\dots,M\} : P_m \neq I \}\bigr|.
\end{equation}
Under local depolarizing noise, each non-identity factor is contracted by $\eta = 1 - \frac{4p}{3}$, so a string of weight $w(P)$ acquires attenuation $\eta^{w(P)}$. The noisy pre-purification state is therefore
\begin{equation}
    \rho \;=\; \mathcal{E}_p(\rho_0)
    \;=\; \frac{1}{2^M}\sum_{P} \eta^{\,w(P)} r_P P .
\end{equation}
Using Pauli orthonormality, we obtain
\begin{align}
    F(\rho)
    = \langle \psi_0 | \rho | \psi_0 \rangle
      &= \frac{1}{2^M}\sum_{P} \eta^{\,w(P)} r_P^2,
    \label{eq:Fp-pauli}\\[4pt]
    \Tr(\rho^2)
      &= \frac{1}{2^M}\sum_{P} \eta^{\,2w(P)} r_P^2 .
    \label{eq:Tp-pauli}
\end{align}
It is convenient to group the squared coefficients by Pauli weight. Define the Pauli-weight distribution of the target
\begin{equation}
    a_k \;:=\; \frac{1}{2^M}\sum_{P:\,w(P)=k} r_P^2,
    \qquad
    \sum_{k=0}^{M} a_k = 1 .
\end{equation}
Then Eqs.~\eqref{eq:Fp-pauli}--\eqref{eq:Tp-pauli} become
\begin{align}
\label{eq:Fidelity_in_eta}
    F(\rho) \;&=\; \sum_{k=0}^{M} a_k\,\eta^{k},\\
    \Tr(\rho^2) \;&=\; \sum_{k=0}^{M} a_k\,\eta^{2k}.
    \label{eq:Fp-Tp-weight}
\end{align}
The coefficients $\{a_k\}$ encode how ``globally'' the target $\ket{\psi_0}$ is supported in Pauli space: states whose weight distribution is biased toward large $k$ suffer faster degradation of both $F(\rho)$ and $\Tr(\rho^2)$ as $p$ increases.

For weak depolarization ($p\ll 1$), we can expand $\eta^k = (1-\tfrac{4p}{3})^k$ to first order, which gives
\begin{align}
    F(\rho) &= 1 - \frac{4}{3}\,\bar{k}\,p + O(p^2),\\
    \Tr(\rho^2) &= 1 - \frac{8}{3}\,\bar{k}\,p + O(p^2),
\end{align}
where
\begin{equation}
\label{eq:smallP_Fin}
    \bar{k} := \sum_k k\,a_k
\end{equation}
is the average Pauli weight of the target $\ket{\psi_0}$.

We now evaluate the action of the PQES map (\ref{eq:purifiedstaten2}). The corresponding output fidelity is
\begin{equation}
\label{eq:VD_Fout_appendix}
    F({\cal P}(\rho))
    \;=\; \frac{\bra{\psi_0}\rho^2\ket{\psi_0}}{\Tr(\rho^2)} .
\end{equation}
To obtain its small-$p$ behavior, write $\rho=\rho_0+\delta\rho$ with $\delta\rho=O(p)$ and $\Tr(\delta\rho)=0$. Then
\begin{align}
    \bra{\psi_0}\rho^2\ket{\psi_0}
    &= \bra{\psi_0}(\rho_0+\delta\rho)^2\ket{\psi_0} \nonumber\\
    &= 1 + 2\bra{\psi_0}\delta\rho\ket{\psi_0} + \bra{\psi_0}\delta\rho^2\ket{\psi_0},
\end{align}
and
\begin{align}
    \Tr(\rho^2)
    &= \Tr\!\bigl[(\rho_0+\delta\rho)^2\bigr] \nonumber\\
    &= 1 + 2\Tr(\rho_0\delta\rho) + \Tr(\delta\rho^2) \nonumber\\
    &= 1 + 2\bra{\psi_0}\delta\rho\ket{\psi_0} + \Tr(\delta\rho^2),
\end{align}
where we used $\rho_0=\ket{\psi_0}\!\bra{\psi_0}$ and $\Tr(\rho_0\delta\rho)=\bra{\psi_0}\delta\rho\ket{\psi_0}$.
Therefore the \emph{linear} terms in $\delta\rho$ cancel in the ratio~\eqref{eq:VD_Fout_appendix}, giving
\begin{align}
    F({\cal P}(\rho))
    &= 1 - \Bigl(\Tr(\delta\rho^2)-\bra{\psi_0}\delta\rho^2\ket{\psi_0}\Bigr) + O(p^3)
    \nonumber\\
    &= 1 - O(p^2).
    \label{eq:smallP_Fout}
\end{align}
In other words, under the PQES map, the leading deviation of the purified fidelity is \emph{quadratic} in the physical error rate:
the $O(p)$ infidelity present in $F(\rho)$ is removed by normalization in~\eqref{eq:VD_Fout_appendix}, and the remaining contribution is governed by the second-order term $\delta\rho^2$ (with a state- and channel-dependent coefficient). This is consistent with the spectral picture that ${\cal P}(\rho)$ effectively squares eigenvalues and renormalizes, producing superlinear convergence in the high-fidelity regime.

\section{Anisotropic errors}
\label{app:axis-dependence}

When considering the behavior of PQES in the context of depolarization vs. dephasing errors, distinct behavior arises. Consider a single qubit with target Bloch unit vector $\hat n=(\sin\theta\cos\phi,\sin\theta\sin\phi,\cos\theta)$. Local $Z$-dephasing with error probability $p_z$ contracts the transverse components by $\beta_z:=1-2p_z$, giving
\begin{equation}
    {\vec r} =(\beta_z \sin\theta\cos\phi,\ \beta_z \sin\theta\sin\phi,\ \cos\theta).
\end{equation}
Two consequences follow:

(i) \emph{Tilt toward the dephasing axis.} Unless $\theta\in\{0,\pi\}$, the Bloch vector of the noisy state tilts toward the dephasing axis.

(ii) \emph{PQES cannot restore erased transverse coherence.} Under the PQES update (\ref{eq:purifiedstaten2}), a single qubit state's Bloch sphere vector updates as (\ref{eq:qubitradialupdate}).  This preserves direction and only stretches the \emph{existing} components. If $X/Y$ coherence was strongly suppressed, the map cannot “rotate” amplitude back into $X/Y$; it can only scale up what is left. Quantitatively, with $F=\frac12(1+\hat n  \cdot {\vec r} )$ and $\text{Tr} (\rho^2) =\frac12(1+\|{\vec r} \|^2)$,
\begin{align}
    \hat n \!\cdot\!{\vec r} &=\beta_z\sin^2\theta+\cos^2\theta
        =1-(1-\beta_z)\sin^2\theta \\
    \|{\vec r} \|^2&=\beta_z^2\sin^2\theta+\cos^2\theta .
\end{align}
After one (effective) purification round,
\begin{equation}
    F'
    =\frac12\!\left(1+\frac{2}{1+\|{\vec r} \|^2}\,[\,\beta_z\sin^2\theta+\cos^2\theta\,]\right).
\end{equation}
For equatorial targets ($\theta=\pi/2$), the input coherence is just $\beta_z$, and the per-round gain is linear:
\begin{align}
    F=\tfrac12(1+\beta_z)\quad\longrightarrow\quad
    F' &=\tfrac12\!\left(1+\frac{2\beta_z}{1+\beta_z^2}\right) \\
    &=\tfrac12\bigl(1+2\beta_z+O(\beta_z^3)\bigr) . 
\end{align}
When $\beta_z\ll 1$ (strong dephasing) each round increases the transverse coherence by at most a constant factor. The core limitation remains: anisotropic noise removes or suppresses particular components of the target Bloch vector, while the PQES power map only amplifies the eigenbasis already present in the noisy state. It cannot rotate the state back toward a target direction whose coherence has been erased.

\end{document}